\documentclass[letterpaper,journal]{IEEEtran}
\usepackage{amssymb}
\usepackage{amsmath,amsfonts}
\usepackage{amsthm}
\usepackage{algorithmic}
\usepackage{color}
\usepackage{algorithm}
\usepackage{multirow}
\usepackage{mathrsfs}
\usepackage{array}
\usepackage{textcomp}
\usepackage{stfloats}
\usepackage{float}
\usepackage{booktabs}
\usepackage{verbatim}
\usepackage{url}
\usepackage{graphicx}
\usepackage{subfigure}
\usepackage{cite}

\begin{document}
\title{ISAC for AI: A Trade-off Framework Across Data Acquisition and Transfer in Federated Learning}

\author{Lai Jiang,~\IEEEmembership{Student Member,~IEEE,} Kaitao Meng,~\IEEEmembership{Member,~IEEE,} Murat Temiz,~\IEEEmembership{Member,~IEEE,} and Christos Masouros,~\IEEEmembership{Fellow,~IEEE}
        
\thanks{Lai Jiang, Murat Temiz and Christos Masouros are with the Department of Electronic and Electrical Engineering, University College London, London, UK (emails: \{lai.jiang.22, m.temiz, c.masouros\}$@$ucl.ac.uk).\\ 

Kaitao Meng is with the School of
Electrical and Electronic Engineering, The University of Manchester, Manchester, UK (email: kaitao.meng$@$manchester.ac.uk).}}

\maketitle

\newtheorem{lemma}{\bf Lemma}
\newtheorem{remark}{\bf Remark}
\newtheorem{Pro}{\bf Proposition}
\newtheorem{theorem}{\bf Theorem}
\newtheorem{Assum}{\bf Assumption}
\newtheorem{Cor}{\bf Corollary}

\begin{abstract}
In this paper, we propose a resource allocation framework for federated learning (FL) in integrated sensing and communication (ISAC) systems, where we consider not only the reliability of model transfer through communication, but also the quality of data acquisition through sensing in the first place. Unlike existing works that assume training data is pre-collected or only impose a fixed sensing signal-to-noise ratio (SNR) threshold to reflect data quality, we explicitly characterize the relationship between sensing data quality (measured by sensing SNR), dataset size, and the upload reliability in FL training, and exploit this relationship to allocate resources between sensing and communication under a shared energy budget.
This is non-trivial due to the intricate coupling among sensing data quality, transmission reliability, and communication resource allocation; nevertheless, it enables a principled joint optimization framework that directly enhances learning performance.
Specifically, we derive a closed-form convergence upper bound that quantifies the joint impact of these factors on the FL optimality gap.
Utilizing this upper bound, the original intractable optimization problem can be reformulated into a tractable resource allocation problem that jointly optimizes the sensing transmit power, number of sensing snapshots, and communication transmit power at each device subject to individual energy budget constraints. 
To solve the reformulated problem, we propose a two-layer optimization algorithm with linear complexity, where the outer layer employs golden section search and the inner layer solves per-device subproblems with closed-form solutions.
Numerical simulations on multiple datasets, including MNIST dataset and the real-world DeepSense 6G radar dataset, demonstrate that the proposed joint optimization strategy consistently outperforms baseline schemes that separately prioritize sensing or communication resources.

\end{abstract}


\section{Introduction}
The upcoming sixth generation (6G) wireless networks are envisioned to support a wide range of intelligent applications, including autonomous vehicles, extended reality (XR), and smart infrastructure. 
Recently, the International Telecommunication Union (ITU-R) has identified integrated sensing and communication (ISAC) and artificial intelligence (AI) as two of the three new usage scenarios for IMT-2030 \cite{singh2025towards}.
These emerging scenarios impose stringent requirements on future networks, including ultra-high data rates, low latency, and strong AI capabilities deployed at the network edge \cite{yang2020artificial, letaief2021edge}.
To meet these requirements, deep learning models need to be trained on large volumes of data collected from distributed devices, which serve as both sensors and communication endpoints in future networks. The conventional approach of centralizing all data at a server for model training faces fundamental limitations, as transmitting raw data from a massive number of devices incurs prohibitive communication overhead and raises serious privacy concerns \cite{sun2019application}. 
These challenges have motivated a new learning paradigm that can leverage the data distributed across devices without requiring its centralization, motivating the adoption of federated learning (FL) as a key enabler of distributed intelligence in 6G networks \cite{bonawitz2019towards}. 

\IEEEpubidadjcol
FL achieves this by enabling a large number of devices to collaboratively train a shared model using their local data, without exchanging raw data with a central server, thereby also preserving data privacy \cite{mcmahan2017communication}. 
Despite the promise, deploying FL over wireless networks introduces unique challenges that are not present in conventional distributed computing settings. In particular, the unreliability of wireless channels leads to stochastic model upload failures, which cause devices to be dropped from the aggregation process and introduce bias into the global model update \cite{ang2020robust, wei2022federated}. Moreover, the heterogeneity of devices in terms of channel conditions, computational capability, and energy budgets makes it difficult to ensure consistent participation of all distributed devices across model upload rounds.
Furthermore, in practical wireless systems, the local data at each device is often generated from heterogeneous environments, resulting in non-independently and identically distributed (Non-IID) data across devices. Such data heterogeneity causes the local gradients to diverge from the global gradient, further slowing convergence and making the aggregation of diverse local updates particularly critical \cite{zhao2021federated}.
In response to these challenges, considerable research effort has been devoted to optimizing communication resources for improving FL performance and reducing communication overhead \cite{chen2020joint, yang2020energy, wan2021convergence, vu2020cell, shi2020joint, huang2022wireless}. 
For instance, \cite{chen2020joint} derives a convergence upper bound for FL over wireless networks and optimizes transmit power and bandwidth allocation to accelerate convergence. Building on this, \cite{huang2022wireless} further proposes joint beamforming and device scheduling to improve FL training performance in MIMO networks.
While valuable insights into the communication aspects of wireless FL have been provided, most of these works treat local training data as fixed and pre-collected. As a result, the role of data acquisition, as well as the impact of data quality and dataset size on FL training performance remains largely overlooked.

Nevertheless, the data that the distributed agents use to learn needs to be acquired first, making  the quality of \textit{{data acquisition}} in the first place a key performance differentiator. In practical systems, training data can be acquired through active sensing, and the resulting data quality is directly governed by the allocated sensing resources. Furthermore, in ISAC deployments, sensing and communication compete for the same energy and spectrum resources, making the quality of sensed data and the reliability of model uploads inherently coupled, which cannot be addressed by optimizing communication resources alone.
In this context, ISAC emerges as a promising technique to tackle this challenge, as it enables devices to simultaneously perform sensing and communication over shared hardware and spectrum \cite{liu2020joint, liu2022integrated, meng2024cooperative1, meng2024cooperative2}.
The resource allocation problem in ISAC systems has been extensively studied, with existing works focusing on optimizing the dual-function waveform design, power allocation, and beamforming to balance sensing and communication performance \cite{liu2021cramer, liuxiang2020joint, Hua10086626}. Building on these foundations, the integration of ISAC with FL opens up new opportunities for data-driven learning at the network edge \cite{liu2022vertical, liu2024multi, zhang2024federated}, and results in entirely new synergies and trade-offs across sensing and communications, with learning-oriented objectives.

\IEEEpubidadjcol
By exploiting the dual functionality of ISAC, devices can actively collect sensing data from the surrounding environment while participating in communication tasks, which makes ISAC a natural fit for FL frameworks where devices are required to both gather training data and upload model updates. This integration gives rise to a new paradigm, where devices sense the environment to acquire local training samples and collaboratively train a model under the coordination of a central server. In such a system, the sensing resource allocation directly determines the quality and quantity of locally collected training data, which in turn affects the convergence of the FL algorithm. This introduces an underexplored dimension to the resource allocation problem in wireless FL \cite{qin2021federated}. 
Unlike conventional studies, where local datasets are usually assumed readily available, the considered ISAC-enabled FL system makes data acquisition part of the resource allocation problem. As a result, sensing, communication, and learning become tightly coupled through the shared energy budget and their joint impact on FL convergence.
Taking a step further than conventional wireless FL works that focus solely on communication resources, a number of works have incorporated sensing constraints into the FL resource allocation framework \cite{hu2025differentially, du2025flisc, liu2022toward}.
Specifically, \cite{hu2025differentially} and \cite{du2025flisc} jointly optimize sensing and communication resources in FL systems and explicitly take into account the fact that the number of locally collected training samples is determined by the sensing resource allocation. However, both works characterize sensing quality by requiring the sensing SNR at each device to exceed a predefined threshold, which is treated as a non-optimized system parameter rather than a variable that quantitatively affects the learning performance. While such a constraint prevents severe sensing degradation, it does not establish the analytical relationship between the sensing SNR and the FL convergence behavior. 
More importantly, the sensing SNR threshold only ensures feasibility of sensing quality, but cannot reflect the effect of sensing performance on the bottom-line learning performance. As a consequence, the resulting resource allocation may still be far from optimal in terms of learning performance even when the sensing constraint is satisfied.
In \cite{liu2022toward}, the impact of sensing transmit power on FL training performance is analyzed theoretically for a human motion recognition task, and the results show that satisfactory learning performance can be achieved once the sensing transmit power exceeds a certain threshold. However, this conclusion is derived under a task-specific signal model and may not generalize to other sensing tasks or data modalities. As a result, the trade-off between sensing and communication cannot be fully characterized under such formulations, and the allocated resources may still be suboptimal from the perspective of learning performance. 
To address this gap, a principled analytical framework that explicitly models the relationship between sensing SNR, data quality, and FL convergence is needed for general learning tasks, which forms the basis of the present work.

\IEEEpubidadjcol
In this paper, we propose an ISAC-enabled federated learning framework in which devices simultaneously sense the environment to collect local training data and upload model updates to the central server. The main contributions of this work are summarized as follows:
\begin{itemize}
    \item We develop a theoretical connection between sensing data quality, dataset size, model upload quality and FL convergence by deriving a closed-form upper bound on the optimality gap that explicitly captures the joint effect of sensing resource allocation, dataset size, and upload reliability on learning performance. This provides a novel analytical framework that links sensing resource allocation directly to FL convergence in an ISAC system. In particular, we model the gradient noise variance as inversely proportional to both the sensing SNR and the local dataset size, and empirically validate this relationship on the MNIST dataset.
    \item Based on the convergence analysis, we reformulate the original intractable FL performance optimization problem into a tractable resource allocation problem that jointly optimizes the sensing transmit power, number of sensing snapshots, and communication transmit power at each device under individual energy budget constraints.
    \item We propose a two-layer optimization algorithm to solve the reformulated problem, where the outer layer applies the golden section search to find the optimal minimum sensing SNR, and the inner layer solves per-device subproblems with closed-form solutions. The overall complexity of the proposed algorithm scales linearly with the number of devices and the maximum number of snapshots, confirming its computational efficiency.
\end{itemize}
Numerical simulations across multiple datasets, including the MNIST dataset and the real-world DeepSense 6G radar dataset, confirm that the proposed framework achieves consistent performance gains over baseline schemes under both IID and Non-IID data distributions, validating the effectiveness of joint sensing and communication optimization.
    
\section{System Model}
\IEEEpubidadjcol
In the considered FL framework, one BS coordinates a set $\mathcal{K}$ consisting of $K$ devices to jointly train a sensing-related deep learning model, using their locally acquired datasets $\mathcal{D}_k, k=1,2,...,K$, where the size of each dataset is $D_k$. The model is trained using FedAvg\cite{mcmahan2017communication}. We define the number of global iterations to achieve a pre-defined accuracy $\varepsilon$ as $E_G(\varepsilon)$.
In each global round, the central server broadcasts the current global model to all devices, each of which then performs $E_L$  local SGD updates \cite{amari1993backpropagation} using its local dataset before uploading the updated model parameters back to the server for aggregation.
The goal of the training framework is to find a global model parameter $\mathbf{w}$ that generalizes well across the heterogeneous local datasets of all participating devices, which is achieved by minimizing the following global loss function \cite{li2020federatedsurvey}
\begin{equation}
    F(\mathbf{w})=\sum^{K}_{k=1}\rho_kF_k(\mathbf{w}),
\end{equation}
where $F_k(\mathbf{w})$ is the local loss function of device $k$ and $\rho_k$ is the weighting factor, where $\sum^K_{k=1}\rho_k=1$. The considered scenario is shown in Fig. \ref{fig_system}, and the underlying sensing-communication trade-off is conceptually illustrated in Fig. \ref{fig_concept}.

\begin{table}[!t]
\caption{Notations and Definitions\label{tab:notations}}
\centering
\scalebox{0.95}{
\begin{tabular}{cc}
\toprule
Notation & Definition\\
\midrule
$K,\mathcal{K}$ & Number and set of devices\\

$E_G$, $E_L$ & Number of global iterations and local iterations\\

$\mathcal{D}_k,D_k$ & Local dataset and its size at device $k$\\

$F(\mathbf{w}),F_k(\mathbf{w})$ & Global and local loss function\\

$\xi$ & random training sample \\

$\mathbf{w}^k_t$ & The local model of device $k$ at round $t$\\

$T^s$ & The sensing duration\\

$T^c_k$ & The time consumption for uploading at device $k$ \\

$T^{train}_k$ & The time consumption for local training at device $k$\\

$T_{th}$ & Delay constraint for each upload round \\

$\tau_k$ & Length of one snapshot at device $k$\\

$L^s_k$ & Number of snapshots at device $k$\\

$\gamma^s_k$, $\gamma^c_k$ & Sensing SNR and uploading SNR at device $k$\\

$p^s_k,p^c_k$ & Sensing and communication transmit power at device $k$\\

$d_k$ & Distance from device $k$ to the target\\

$d_{k,s}$ & Distance from device $k$ to the server\\

$G_k$ & Antenna gain of device $k$\\

$\psi_k, \Omega^s_k$ & Small-scale fading factor and its mean\\

$E^{train}_k$ & The energy consumption for training at device $k$ \\

$E_k$ & Energy budget at device $k$ \\

$\kappa_k$ & Energy consumption per sample at device $k$ \\

$h_k, \Omega^c_k$ & channel gain factor and its mean\\


$\alpha_p$ & Large-scale pathloss factor \\

$z$ & Size of model parameters in bits \\

$R_k$ & Data rate at device $k$ \\

$P^{out}_k$ & Uploading outage probability at device $k$ \\

$q_k$ & The probability of successful uploading at device $k$ \\

$\rho^{(t)}_k$ & Aggregation weighting factor of device $k$ \\

$S_t$ & The active set of devices \\

$\zeta^2$, $\sigma^2_0$ & Upper-bound constant and gradient noise constant \\
\bottomrule
\end{tabular}}
\end{table}

\subsection{Environmental Sensing for Data Collection}
During the sensing stage, each device $k$ transmits a signal with power $p^s_k$ to sense the environment and obtain sensing information within a range of $d_k\in[d_\text{min}, d_\text{max}]$. The sensing transmit power should be larger than a threshold to ensure the quality of sensed data, and less than the maximal transmit power
\begin{equation}
    p^s_\text{min} \leq p^s_k \leq p^s_\text{max}, k=1,2,...,K.
\end{equation}

We assume the length of one snapshot is $\tau_k$ at device $k$, and $L^s_k$ is the number of snapshots to obtain a single sensing frame. Then the required time to obtain a sensing frame is $T_{k,0}=L^s_k\tau_k,\ k=1,2,...,K$.
During the sensing stage $T^s$, the total number of sample collection attempts is calculated as $D_k=T^s/T_{k,0}=T^s/(L^s_k\tau_k)$.
To maintain the temporal validity of acquired samples, the time consumption for sensing will not exceed a given duration, which is expressed by
\begin{equation}
    \tau_k \leq T^s \leq T^s_{th}, k=1,2,...,K.
\end{equation}

Coherent accumulation is adopted in the considered system, where multiple snapshots are integrated to improve the sensing signal-to-noise ratio (SNR). In this case, the effective SNR can be strengthened by increasing both sensing transmit power $p^s_k$ and the number of snapshots $L^s_k$ at the cost of increased energy consumption and reduced amount of available training samples. The sensing SNR at device $k$ is given by
\begin{equation}
\label{eq:sensing_snr}
    \gamma^s_{k}=\frac{p^s_kG_k L^s_k\psi_{k}}{\sigma^2_sd^{\alpha_p}_{k}},
\end{equation}
where $G_k$ is the antenna gain, $\psi_{k}$ is a random variable quantifying the small-scale fading, which follows an exponential distribution, i.e., $\psi_{k} \sim \text{Exp}(\frac{1}{\Omega^s_k})$ with $\Omega^s_k$ being the mean. The power of Gaussian noise is denoted by $\sigma^2_s$, $d_k\in[d_\text{min}, d_\text{max}]$ is the distance and $\alpha_p$ models the effect of the round-trip pathloss. The sensing SNR expression in \eqref{eq:sensing_snr} reveals that improving data quality and increasing dataset size impose competing demands on energy budget, as allocating more power or snapshots to sensing enhances data quality but leaves less energy for communication. These coupled dependencies between sensing quality, dataset size, and communication reliability highlight the need for a principled joint optimization framework.

\begin{figure*}[htbp]
\setlength{\abovecaptionskip}{0.cm}
	\centering
	\begin{minipage}{0.49\linewidth}
		\centering
		\includegraphics[scale=0.07]{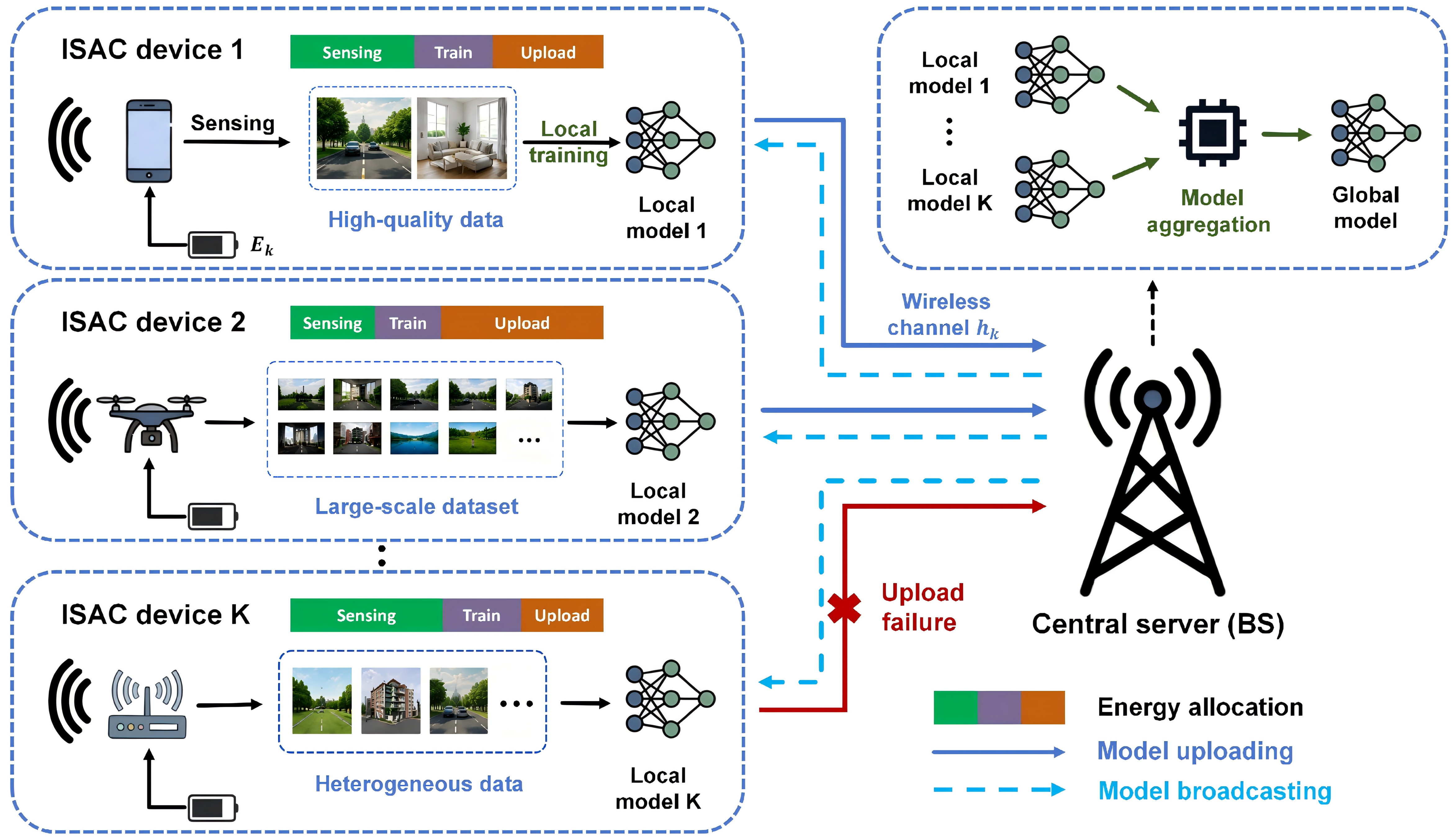}
		\caption{The illustration of ISAC-enabled federated learning framework.}
		\label{fig_system}
	\end{minipage}
	\begin{minipage}{0.49\linewidth}
		\centering
		\includegraphics[scale=0.07]{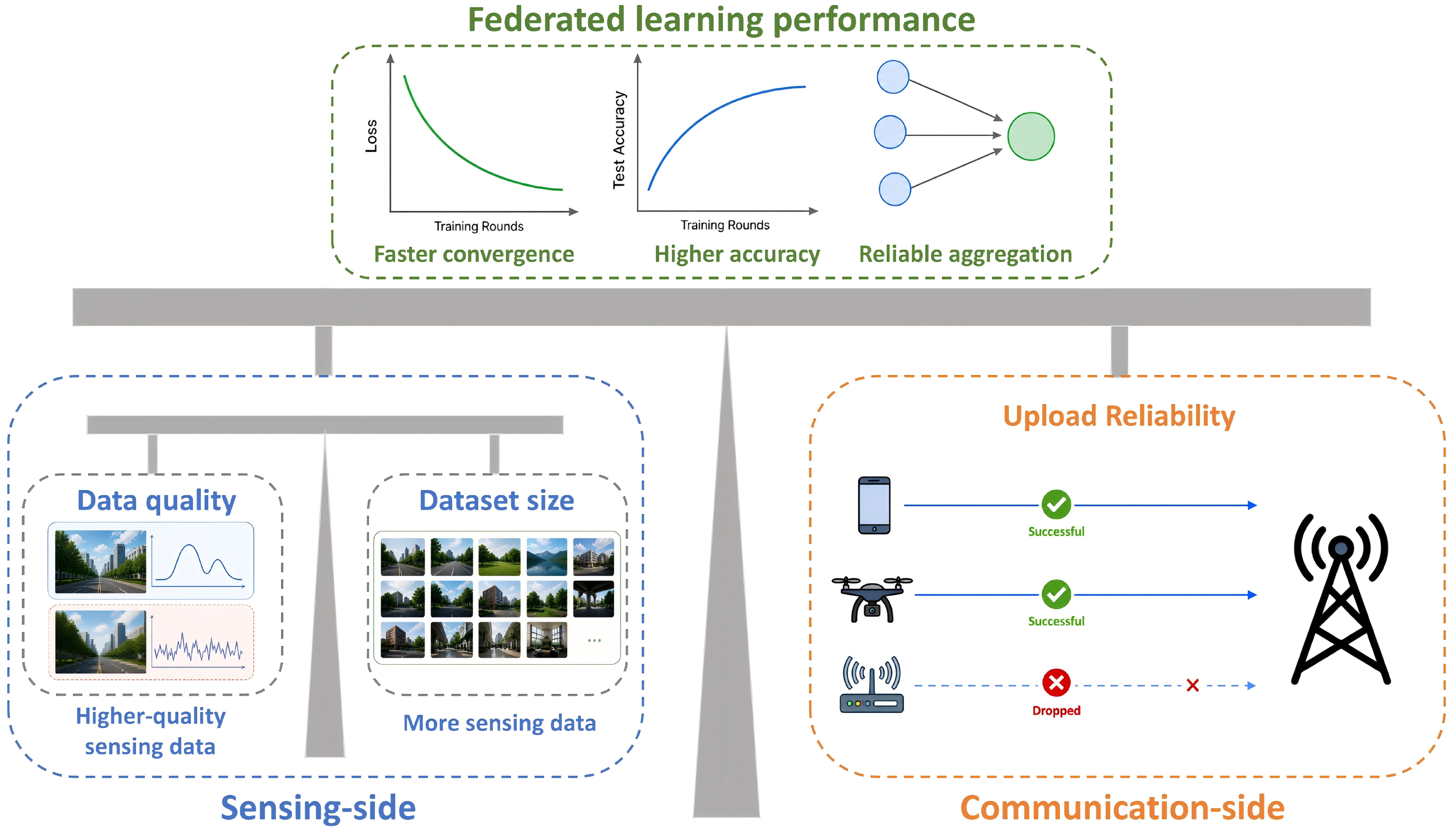}
		\caption{Sensing-communication trade-off in ISAC-enabled federated learning.}
		\label{fig_concept}
	\end{minipage}
    \vspace{-12pt}
\end{figure*}

\subsection{Models broadcast and Local Training}
\IEEEpubidadjcol
At the start of each global epoch $t \in \mathcal{I}_g=\{nE_L\mid n=1,2,...\}$, the central server broadcasts the global model parameters $\mathbf{w}_t$ to each participating device. We assume the latency for the downlink transmission can be ignored as the server has greater transmit power than devices to support high-rate transmission \cite{liu2018toward}. 
After receiving the global model parameters, the local device $k$ lets $\mathbf{w}^k_t=\mathbf{w}_t$, and performs $E_L$ local epochs to update the parameters as
\begin{equation}
    \mathbf{w}_{t+i+1}^k=\mathbf{w}_{t+i}^k -\eta_{t+i}\nabla F_k(\mathbf{w}_{t+i}^k;\xi),
\end{equation}
where $i=0,...,E_L-1$, $\eta_{t+i}$ denotes the learning rate, and $\nabla F_k(\mathbf{w}_{t+i}^k;\xi)$ is the stochastic gradient on sampled data $\xi$.


\subsection{Uploading and Model Aggregation}
\IEEEpubidadjcol
After model training, the devices upload the model parameters to the central server for model aggregation. The size of parameters to be transmitted is denoted by $z$ in bits. We assume block fading channels and the device $k$ transmits with power $p^c_k$, the instantaneous uplink SNR is given by
\begin{equation}
    \gamma^c_k=\frac{p^c_kG_kh_k}{\sigma^2_c d_{k,s}^{\alpha_p/2}},
\end{equation}
where $d_{k,s}$ is the distance from device $k$ to the server, $h_k$ is a random variable denoting the channel gain, which is assumed to follow $h_k\sim \text{Exp}(\frac{1}{\Omega^c_k})$ with $\Omega^c_k$ being the mean. The uploading rate is given by
\begin{equation}
    R_k=B_k\log_{2}(1+\gamma^c_k),
\end{equation}
where $B_k$ is the allocated bandwidth for device $k$. Therefore, the required time for uploading at device $k$ is $T^c_k=z/R_k$. 
The total required time for training and uploading should satisfy a delay constraint $T_{th}$ as excessively long transmission delays would slow down the global training process, so we have
\begin{equation}
    T^{train}_k+T^c_k \leq T_{th}, k=1,2,...,K.
\end{equation}
Given the delay constraint, the required minimum transmission data rate is $R^{req}_k=z/(T_{th}-T^{train}_k)$.
A successful transmission occurs if the uploading data rate satisfies $R_k \geq R^{req}_k$.
Otherwise, the transmission fails due to outage and the device is regarded as dropped in that training round. And the dropout probability of device $k$ can be given by
\begin{equation}
\label{eq:p_out}
\begin{aligned}
    P^{out}_k &= \rm Pr\{ \log_{2}(1+\gamma^c_k) \leq \it R^{req}_k/B_k \} \\ &=\rm 1-exp\Bigg(- \frac{\it \sigma^2_c d_{k,s}^{\alpha_p/2}}{\it p^c_k G_k \rm \Omega \it ^c_k}\Big(2^{\it R^{req}_k/B_k}-1\Big)\Bigg),
    \end{aligned}
\end{equation}
which directly impacts the global aggregation and convergence performance of the learning process. Therefore, the probability of successful uploading is $q_k = 1-P^{out}_k$. The maximum transmit power for uploading is limited by
\begin{equation}
    p^c_\text{min} \leq p^c_k \leq p^c_\text{max}, k=1,2,...,K.
\end{equation}

In this work, we assume that the upload success events of different devices are independent in each round and are modeled by a Bernoulli random variable $C^{(t)}_k \sim \text{Bernoulli}(q_k)$. $C^{(t)}_k \in \{0,1\}$ indicates whether device $k$ successfully uploads at iteration $t$, where
\begin{equation}
    C^{(t)}_k=
    \begin{cases}
        1, \quad \text{with probability}\ 1-P^{out}_k, \\
        0, \quad \text{with probability}\ P^{out}_k.
    \end{cases}
\end{equation}

We define $S_t=\{k:C^{(t)}_k=1\}$ as the active set of devices whose updates are successfully received at round $t$, the central server performs global aggregation upon receiving the model parameters from the devices as
\begin{equation}
    \mathbf{w}_{t+E_L}=\sum_{k\in S_t}\rho^{(t)}_k\mathbf{w}_{t+E_L}^k,
\end{equation}
the aggregation weight at the $t$-th iteration is defined as 
\begin{equation}
    \rho^{(t)}_k=
    \begin{cases}
        \frac{D_kC^{(t)}_k}{W_t}, &W_t>0,\\
        0, \quad &W_t=0,
    \end{cases}
\end{equation}
$W_t=\sum_{k \in S_t}D_k=\sum_k D_kC^{(t)}_k$ is the total data weight of participating devices, which can be approximated by $E[W_t]=\sum_k D_k E[C_k^{(t)}]\approx\sum_k D_kq_k$.

\subsection{Energy consumption model}
The total energy budget $E_k$ at each device is consumed across three stages: sensing, local training, and model uploading. It should be noted that the communication energy is consumed regardless of whether the model uploading is successful. The energy budget constraint of device $k$ can be expressed as \cite{wan2021convergence}
\begin{equation}
    p^s_kT^s + E_G(\varepsilon)[p^c_k\hat{T}^c_k + E^{train}_k] \leq E_k, k=1,2,...,K,
\end{equation}
where $\hat{T}^c_k = T_{th}-T_k^{train}$ denotes the maximal allowable transmission delay, and the energy consumption for local computation $E^{train}_k$ is calculated as
\begin{equation}
    E^{train}_k= \kappa_kD_k, \quad \forall k,
\end{equation}
where $\kappa_k$ reflects the computational cost per sample depending on the specific processor on device $k$. Specifically, $\kappa_k$ captures the joint effects of CPU clock frequency, hardware architecture, and the number of CPU cycles required per sample.

\subsection{Problem formulation}
\IEEEpubidadjcol
To jointly optimize the sensing and communication resources for the proposed ISAC-FL framework, we now formulate a resource allocation problem whose goal is to minimize the global loss function $F(\mathbf{w})$ while satisfying energy budget and power constraints at each device as
\begin{subequations}
\label{P:P1}
\begin{align}
    \underset{\{\mathbf{L}^s, \mathbf{P}^s, \mathbf{P}^c\}}{\text{min}}
        & \quad F(\mathbf{w})   \label{eq:cost}\\
    \mbox{s.t.} \quad 
        & 1 \leq L^s_k \leq L^s_\text{max}, \quad \forall k, \label{eq:cost1}\\
        &p^s_kT^s + E_G(\varepsilon)[p^c_k\hat{T}^c_k + E^{train}_k] \leq E_k, \quad \forall k, \label{eq:cost2}\\
        &p^s_\text{min} \leq p^s_k \leq p^s_\text{max}, \quad \forall k, \label{eq:cost3}\\
        &p^c_\text{min} \leq p^c_k \leq p^c_\text{max}, \quad \forall k, \label{eq:cost4}
\end{align}
\end{subequations}
where $\mathbf{L}^s=[L^s_1,L^s_2,...,L^s_K]$, $\mathbf{P}^s=[p^s_1,p^s_2,...,p^s_K]$ and $\mathbf{P}^c=[p^c_1,p^c_2,...,p^c_K]$. Constraint \eqref{eq:cost1} represents the range of available number of snapshots, and \eqref{eq:cost2} is the energy consumption budget of each device. Constraints \eqref{eq:cost3} and \eqref{eq:cost4} are the maximal sensing and communication transmit power constraint. 
The optimization problem aims to minimize the global loss function $F(\mathbf{w})$ by jointly optimizing the sensing and communication resource allocation. However, directly solving \eqref{P:P1} is intractable since $F(\mathbf{w})$ depends on the entire FL training trajectory and has no closed-form expression with respect to the optimization variables. To obtain a tractable objective, we analyze the convergence behavior of the FL algorithm in the following section, which establishes an explicit relationship between the resource allocation variables and the learning performance.

\section{Convergence Analysis of ISAC-FL}
In this section, we analyze the convergence behavior of the proposed ISAC-FL framework to derive a tractable surrogate for the original optimization objective $F(\mathbf{w})$. We first introduce the necessary assumptions, then establish a closed-form upper bound on the FL optimality gap, which explicitly reveals how the sensing and communication resource allocation jointly governs the convergence performance.

\subsection{Assumptions}
\IEEEpubidadjcol
Since the objective in problem \eqref{P:P1} involves the global loss function $F(\mathbf{w})$, which is determined by the entire FL training process, directly solving the original problem \eqref{P:P1} is generally intractable.
Instead, we analyze the convergence behavior of the FL algorithm to derive a tractable surrogate objective.
We first introduce the following assumptions on the loss functions, which are widely adopted in the literature \cite{chen2020joint, hu2025differentially}.

\textbf{Assumption 1.} We assume the loss functions at each device $k=1,2,...,K$ are strongly convex with a positive parameter $\mu$, such that
\vspace{-1.5mm}
\begin{equation}
\begin{aligned}
    F_k(y) \geq F_k(x) + (y-x)^T\nabla F_k(x) + \frac{\mu}{2}\parallel y-x\parallel^2_2, \quad \forall k.
\end{aligned}
\end{equation}

\textbf{Assumption 2.} We assume the local loss functions at each device $k=1,2,...,K$ are $L$-smooth, such that
\vspace{-1.5mm}
\begin{equation}
    \parallel \nabla F_k(y) - \nabla F_k(x) \parallel \leq L \parallel y-x \parallel, \quad \forall k.
\end{equation}
Note that $L$-smooth also holds for the global loss $F(\mathbf{w})$ under the weighted sum of $F_k(\mathbf{w})$.
These two assumptions are standard and hold for a wide range of commonly used loss functions, including mean squared error, logistic regression loss and cross entropy \cite{friedlander2012hybrid}. Beyond the conditions above, the considered ISAC system introduces additional factors that influence convergence, including the quality and quantity of locally sensed training data as well as the reliability of model uploading. To capture the effects of these factors on convergence behavior, we make the following assumptions.

\textbf{Assumption 3.} The variance of stochastic gradient computed on noisy sensing data at device $k$ is upper-bounded by
\begin{small}
\begin{equation}
\label{eq:assumption3}
    E[\parallel \tilde{g}_k(\mathbf{w}^k_t;\xi)-g_k(\mathbf{w}^k_t) \parallel^2] \leq \frac{\sigma_k^2(\gamma^s_k)}{D_k}.
\end{equation}
\end{small}Here, $\sigma^2_k(\gamma^s_k)$ is a sensing SNR-dependent function that characterizes how the quality of locally sensed data influences the gradient estimation error, and its explicit form is derived in Section \ref{sec_sensing_quality_analysis}.

\textbf{Assumption 4.} The gap between the local gradient on each device and global gradient is upper-bounded by \cite{haddadpour2019convergence}
\begin{small}
\begin{equation}
    \sum^K_{k=1}\parallel \nabla F_k(\mathbf{w}) - \nabla F(\mathbf{w}) \parallel^2 \leq \zeta^2,
\end{equation}
\end{small}
where the constant $\zeta^2$ quantifies the degree of non-IID, and $\zeta^2=0$ when the data are IID across devices. This assumption is practically justified in the considered system, where devices sense the same environment from different locations, resulting in heterogeneous yet correlated local datasets.

\subsection{Sensing-Dependent Gradient Noise Analysis}
\label{sec_sensing_quality_analysis}
In the considered ISAC-FL system, the training data at each device is collected through the sensing process, which inevitably introduces AWGN noise into the samples. To characterize this effect analytically, we establish the following analysis, which shows that the gradient noise variance is inversely proportional to both the sensing SNR $\gamma^s_k$ and the local dataset size $D_k$. To facilitate the convergence analysis, we consider the average sensing SNR by taking the expectation over the small-scale fading $\psi_k$, i.e., $E[\psi_k]=\Omega^s_k$, such that $\gamma^s_k=p^s_kG_kL^s_k \Omega^s_k/(\sigma^2_s d^{\alpha_p}_k)$ in the subsequent analysis.
\begin{Pro}
Suppose the gradient of loss function $F(\mathbf{w})$ is $L_x$-Lipschitz continuous with respect to the input sample, i.e., $\parallel \nabla F_k(\mathbf{w}^t_k;\tilde{\xi})-\nabla F_k(\mathbf{w}^t_k;\xi) \parallel \leq L_x \parallel \tilde{\xi}-\xi \parallel$, the upper bound of gradient noise can be denoted by
\begin{small}
\begin{equation}
    \sigma_k^2(\gamma^s_k)=\frac{\sigma_0^2}{\gamma^s_k},
\end{equation}
\end{small}where $\sigma^2_0$ is a constant. Therefore, Eq. \eqref{eq:assumption3} can be rewritten as 
\begin{small}
\begin{equation}
    E[\parallel \tilde{g}_k(\mathbf{w}^k_t;\xi)-g_k(\mathbf{w}^k_t) \parallel^2] \leq \frac{\sigma_0^2}{D_k \gamma^s_k}.
\end{equation}    
\end{small}
\end{Pro}
\begin{proof}
    Please refer to Appendix B.
\end{proof}
The above proposition reflects the following properties:
\begin{itemize}
    \item The devices with low sensing SNR can hinder the convergence of the global model, as low-quality samples can amplify the gradient estimation errors by introducing additional bias into the gradient computation.
    \item Increasing the number of sensed samples can reduce gradient estimation errors, thereby accelerating convergence.
    \item There exists an inherent trade-off between dataset size and data quality in the considered system, as collecting more samples requires fewer snapshots per frame, which degrades the sensing SNR and in turn increases gradient noise during training.
\end{itemize}
\IEEEpubidadjcol
This captures the joint effect of sensing data quality and dataset size on gradient estimation, which directly impacts the convergence of FL as shown in Theorem \ref{theo:convergence}. 
To empirically validate Assumption 3, Fig. \ref{fig_snr_vs_grad_acc} presents the relationship between sensing SNR and gradient variance measured on the MNIST dataset. As shown in \ref{fig_snr_vs_gradvar}, the gradient variance decreases monotonically as sensing SNR increases. Fig. \ref{fig_snr_vs_acc} further confirms that higher sensing SNR leads to improved test accuracy, demonstrating the direct impact of data quality on training. Fig. \ref{fig_loss_curves} illustrates that higher SNR accelerates convergence and yields lower loss, while larger $D_k$ consistently reduces loss across all SNR levels. These observations jointly support the modeling of the gradient noise upper bound as inversely proportional to both $\gamma^s_k$ and $D_k$ in Assumption 3.

\begin{figure}[!t]
\setlength{\belowcaptionskip}{-1cm}
\centering
\subfigure[]{\includegraphics[scale=0.3]{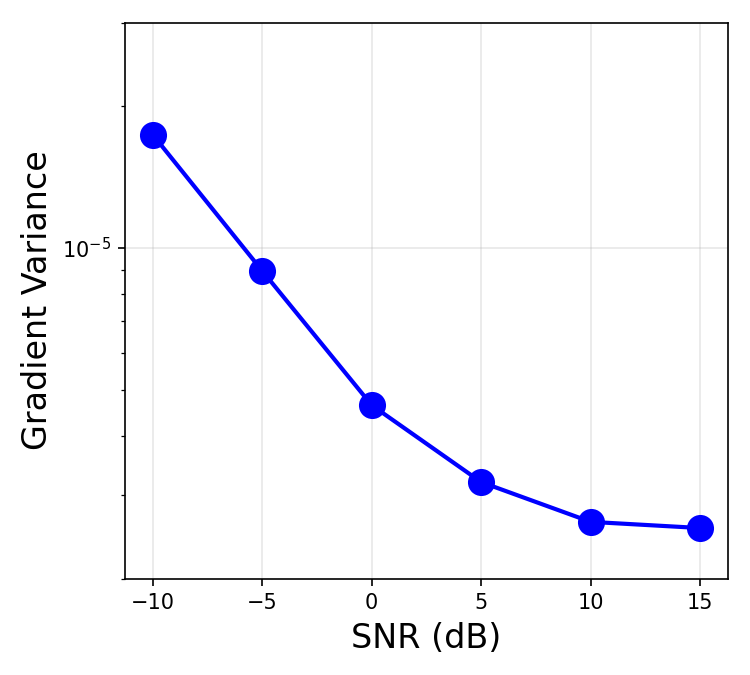}%
\label{fig_snr_vs_gradvar}}
\hfil
\subfigure[]{\includegraphics[scale=0.3]{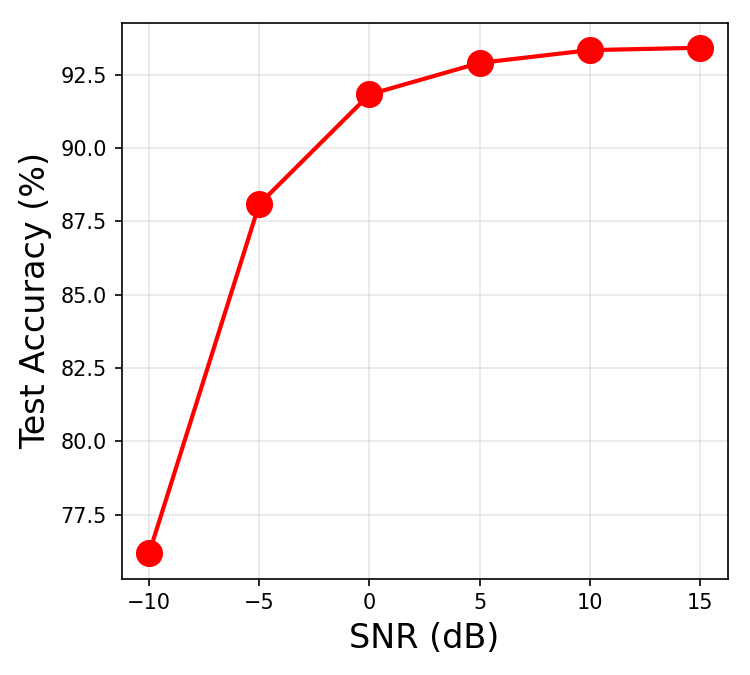}%
\label{fig_snr_vs_acc}}
\caption{Empirical validation of Assumption 3 on the MNIST dataset using a 3-layer MLP network. $(a)$ Gradient variance monotonically decreases with sensing SNR. $(b)$ Test accuracy increases with sensing SNR.}
\label{fig_snr_vs_grad_acc}
\end{figure}

\begin{figure}[!t]
\setlength{\abovecaptionskip}{-0.1cm}
\setlength{\belowcaptionskip}{-1cm}
\centering
\subfigure[$D_k=1000$.]{\includegraphics[scale=0.3]{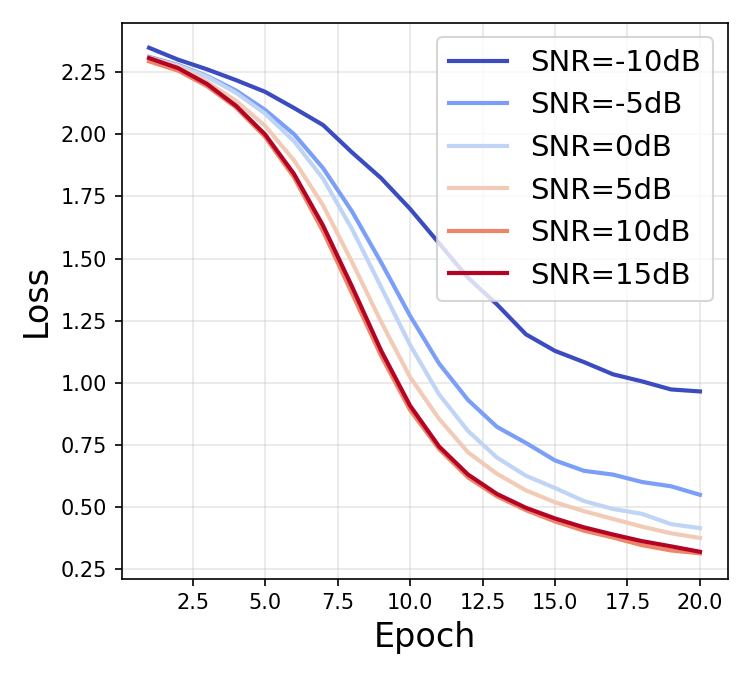}%
\label{fig_loss_curves_D1000}}
\subfigure[$D_k=5000$.]{\includegraphics[scale=0.3]{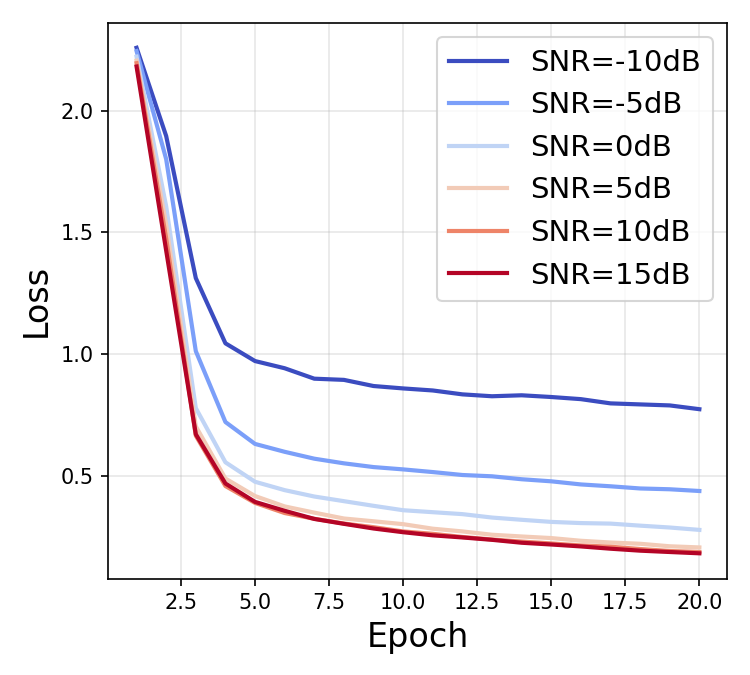}%
\label{fig_loss_curves_D5000}}
\caption{Training loss curves under different sensing SNR levels for varying local dataset sizes $D_k$.}
\label{fig_loss_curves}
\end{figure}

\subsection{Convergence Bound}
\IEEEpubidadjcol
Based on the assumptions above, we now derive a bound on the aggregated gradient noise variance for establishing the convergence bound later in Theorem 1.
\begin{lemma}\label{lm:gradient_noise}
    Let $\varepsilon_t=\tilde{g}_t - g_t$ denote the aggregated gradient noise, the variance of the aggregated gradient noise satisfies
\begin{equation}
    E[\parallel \varepsilon_t \parallel^2] \leq \frac{\sigma^2_0}{\gamma^s_{\text{min}}}\cdot\frac{K}{\sum_k D_kq_k},
\end{equation}
where $\gamma^s_{\text{min}}=\text{min}_k\{\gamma^s_k\}$.
\end{lemma}
\begin{proof}
    Please refer to Appendix C.
\end{proof}
 
\begin{theorem}\label{theo:convergence}
    Given $\eta \leq \frac{1}{2L}$, the expected optimality gap of federated learning after $t$ iterations is bounded by
\begin{small}
\begin{equation}
\begin{aligned}
    E[F(\mathbf{w}_{t})-F^*] &\leq A^t(F(\mathbf{w}_0)-F^*) \\
    &+ \frac{2L\eta^2(1-A^t)}{1-A}\Big(\frac{\sigma^2_0}{\gamma^s_{\text{min}}}\cdot\frac{K}{\sum_k D_kq_k}+ \zeta^2\sum^K_{k=1}\bar{\rho}_k^2 \Big),
\end{aligned}
\end{equation}
\end{small}
where $A=1-2\mu\eta+2L\mu\eta^2$ and $F^*$ is the minimal value of the loss function.
\end{theorem}
\begin{proof}
    Please refer to Appendix D.
\end{proof}

\IEEEpubidadjcol
\subsection{Problem Transformation}
The convergence bound in Theorem 1 reveals that the optimality gap is dominated by the term $\frac{\sigma^2_0}{\gamma^s_\text{min}}\cdot \frac{K}{\sum_k D_k q_k}$, which explicitly captures the joint effect of sensing quality, dataset size, and upload reliability on the learning performance. 
Since the heterogeneity term satisfies $\zeta^2\sum_k \bar{\rho}^2_k \leq \zeta^2\sum_k \bar{\rho}_k =\zeta^2$, we focus on minimizing the dominant component as a tractable surrogate for the original objective $F(\mathbf{w})$.
Accordingly, the optimization problem can be rewritten as
\begin{subequations}
\label{P:P2}
\begin{align}
    \underset{\{\mathbf{L}^s, \mathbf{P}^s, \mathbf{P}^c\}}{\text{min}}
        & \quad \frac{\sigma^2_0}{\gamma^s_{\text{min}}}\cdot\frac{K}{\sum_k D_kq_k},   \\
    \mbox{s.t.} \quad 
        & \gamma^s_\text{min}=\text{min}\{\gamma^s_1,..., \gamma^s_K\},\\
        & 1 \leq L^s_k \leq L^s_\text{max}, \quad \forall k, \\
        &p^s_kT^s + E_G(\varepsilon)[p^c_k\hat{T}^c_k + E^{train}_k] \leq E_k, \quad \forall k, \label{eq:cost_energy_constraint}\\
        &p^s_\text{min} \leq p^s_k \leq p^s_\text{max}, \quad \forall k,\\
        &p^c_\text{min} \leq p^c_k \leq p^c_\text{max}, \quad \forall k.
\end{align}
\end{subequations}
The objective function of problem \eqref{P:P2} is inversely proportional to $\gamma^s_\text{min}=\text{min}_k {\gamma^s_k}$, which means that the convergence performance is bottlenecked by the device with the lowest sensing SNR. Moreover, the objective decreases monotonically as $\sum_k D_k q_k$ increases, which represents the effective aggregated data contribution weighted by the upload success probability of each device. This implies that convergence benefits from both a larger local dataset size $D_k$ and a higher upload success probability $q_k$. However, these two terms are coupled through the energy budget constraint \eqref{eq:cost_energy_constraint}, as allocating more energy to sensing increases $D_k$ and $\gamma^s_k$ at the cost of reducing the power available for communication, and vice versa. The problem \eqref{P:P2}
therefore captures a fundamental sensing-communication trade-off, where the optimal resource allocation must balance data quality and quantity against upload reliability to maximize the overall learning performance.

\section{Proposed Solution}
The reformulated problem remains challenging to solve directly due to the coupling between the minimum sensing SNR constraint and the per-device resource variables. In this section, we propose a two-layer optimization algorithm that decouples the problem into a tractable outer search and a set of independent per-device subproblems.

\subsection{Two-Layer Optimization Framework}
To solve the optimization problem \eqref{P:P2}, we introduce the auxiliary variable $\gamma_t$, and transform the problem to 
\begin{subequations}
\label{P:P3}
\begin{align}
    \underset{\{\gamma_t, \mathbf{L}^s, \mathbf{P}^s, \mathbf{P}^c\}}{\text{max}}
        & \quad \gamma_t \cdot \sum_{k=1}^K D_kq_k,   \\
    \mbox{s.t.} \quad 
    & \gamma_t \leq \gamma^s_k, \quad \forall k, \\
        & 1 \leq L^s_k \leq L^s_\text{max}, \quad \forall k, \\
        &p^s_kT^s + E_G(\varepsilon)[p^c_k\hat{T}^c_k + E^{train}_k] \leq E_k, \quad \forall k,\\
        &p^s_\text{min} \leq p^s_k \leq p^s_\text{max}, \quad \forall k,\\
        &p^c_\text{min} \leq p^c_k \leq p^c_\text{max}, \quad \forall k.
\end{align}
\end{subequations}
It can be proved that \eqref{P:P2} and \eqref{P:P3} are equivalent. Specifically, at optimality the constraint $\gamma_t \leq \gamma^s_k $ must be tight for at least one device, since any $\gamma^*_t < \text{min}_k\{\gamma^{s*}_k\}$ could be strictly increased without violating feasibility, thereby improving the objective. Hence $\gamma^*_t = \text{min}_k\{\gamma^{s*}_k\}$ always holds at the optimal solution.
As problem \eqref{P:P3} remains non-convex due to the coupling between $\gamma_t$ and the remaining variables, we propose a two-layer framework that decouples the problem into a set of independent subproblems. For a fixed $\gamma_t$, we define
\begin{equation}
    W(\gamma_t) \triangleq \text{max} \sum_k D_k q_k, \quad A_k \triangleq \frac{G_k \Omega_k^s}{\sigma^2_sd^{\alpha_p}_k}.
\end{equation}
The problem then reduces to
\begin{equation}
    \underset{\gamma_t \in [\gamma_{lb},\gamma_{ub}]}{\text{max}} \gamma_t \cdot W(\gamma_t),
\end{equation}
where the search bounds are decided by the constraints
\begin{equation}
    \gamma_{lb}=p^s_\text{min}\cdot \text{min}_k \{A_k\}, \quad \gamma_{ub}=\text{min}_k \{A_kp^s_\text{max} L^s_\text{max}\}.
\end{equation}
\IEEEpubidadjcol
We observe that when $\gamma_t$ is small, the SNR constraint $A_k p^s_k L^s_k \geq \gamma_t$ can be easily satisfied, leaving sufficient energy for communication so that $W(\gamma_t)$ remains large and $\gamma_t W(\gamma_t)$ increases approximately linearly. As $\gamma_t$ increases, the SNR constraint tightens, forcing higher sensing power and reducing the energy available for communication, which causes $W(\gamma_t)$ to decrease. When the value of $\gamma_t$ approaches $\gamma_{ub}$, $\gamma_t W(\gamma_t)$ decreases monotonically dominated by the rapid decline of $W(\gamma_t)$. Based on the observation, we apply the golden section search method \cite{hendrix2010introduction} to efficiently find the optimal $\gamma^*_t$. At each iteration, we let
\begin{equation}
    \gamma_1=b-\frac{b-a}{\varphi}, \quad \gamma_2=a+\frac{b-a}{\varphi},
\end{equation}
where $\varphi=(1+\sqrt{5})/2$ is the golden ratio and $[a,b]$ is the current search interval. The interval is updated as
\begin{equation}
    [a,b]=
    \begin{cases}
        [\gamma_1,b], & \gamma_1 W(\gamma_1) < \gamma_2 W(\gamma_2) ,\\
        [a, \gamma_2], \quad &\text{otherwise},
    \end{cases}
\end{equation}

\subsection{Per-Device Subproblem Solution}
\IEEEpubidadjcol
For a given $\gamma_t$, the problem can be decomposed into $K$ independent subproblems. This follows from the observation that both the objective $\sum_k D_k q_k$ and all constraints are separable across devices without coupling. The overall optimal value therefore satisfies $W(\gamma_t)=\sum_k W_k(\gamma_t)$, where $W_k(\gamma_t)$ is the optimal value of the $k$-th subproblem
\begin{subequations}
\label{P:P4}
\begin{align}
    \underset{\{L^s_k, p^s_k, p^c_k\}}{\text{max}}
        & \quad D_kq_k,   \\
    \mbox{s.t.} \quad 
    & A_k p^s_k L^s_k \geq \gamma_t, \quad \forall k, \label{P:P4_snr_constrain}\\
        & 1 \leq L^s_k \leq L^s_\text{max}, \quad \forall k, \\
        &p^s_kT^s + E_G(\varepsilon)[p^c_k\hat{T}^c_k + E^{train}_k] \leq E_k, \quad \forall k,\\
        &p^s_\text{min} \leq p^s_k \leq p^s_\text{max}, \quad \forall k,\\
        &p^c_\text{min} \leq p^c_k \leq p^c_\text{max}, \quad \forall k.
\end{align}
\end{subequations}
For each subproblem \eqref{P:P4}, we fix $L^s_k$ and derive the optimal $p^{s*}_k$ and $p^{c*}_k$ in closed form. Given $L^s_k$, the SNR constraint \eqref{P:P4_snr_constrain} imposes a lower bound on $p^s_k$ as
\begin{equation}
    p^s_k \geq \frac{\gamma_t}{A_k L^s_k}.
\end{equation}
Since the term $D_kq_k$ does not depend explicitly on $p^s_k$, and increasing the value of $p^s_k$ beyond its lower bound reduces the energy available for communication, the optimal sensing power can be given by
\begin{equation}
\label{eq:optimal_ps}
    p^{s*}_k= \text{max} \{p^s_\text{min}, \frac{\gamma_t}{A_k L^s_k} \}.
\end{equation}
Then given the optimal $p^{s*}_k$, substituting $D_k=T^s/(L^s_k\tau_k)$ into the energy constraint yields the maximum allowable communication power, as
\begin{equation}
    p^c_k \leq \frac{E_k-p^{s*}_k T^s-E_G \kappa_k D_k}{E_G \hat{T}^c_k}.
\end{equation}
Since $q_k$ is strictly increasing in $p^c_k$, the optimal strategy for uploading exhausts the remaining energy budget, given by
\begin{equation}
\label{eq:optimal_pc}
    p^{c*}_k=\text{min} \{p^c_\text{max}, \frac{E_k-p^{s*}_k T^s-E_G \kappa_k D_k}{E_G \hat{T}^c_k} \}.
\end{equation}
As $L^s_k$ is a bounded integer variable, the global optimum of problem \eqref{P:P4} can be obtained by full enumeration over all candidate values of $L^s_k$. The complete procedure is summarized in \textbf{Algorithm \ref{alg:alg1}}.

\subsection{Complexity Analysis}
\IEEEpubidadjcol
The computational complexity of the proposed algorithm is analyzed in this subsection. The golden section search requires $O(\log(1/\varepsilon))$ iterations to converge to within tolerance $\varepsilon$, since each iteration reduces the search interval by a constant factor of $1/\varphi$. At each outer iteration, Algorithm \ref{alg:alg2} is executed for each of the $K$ devices independently, where the inner enumeration over $L^s_k$ incurs a cost of $O(L^s_\text{max})$ per device, with the optimal $p^{s*}_k$ and $p^{c*}_k$ obtained in closed form via \eqref{eq:optimal_ps} and \eqref{eq:optimal_pc} at each step. The total inner-layer complexity per outer iteration is therefore $O(K\cdot L^s_\text{max})$. Combining both layers, the overall complexity of Algorithm \ref{alg:alg1} is $O(\log(1/\varepsilon)\cdot K \cdot L^s_\text{max})$. It is worth noting that the complexity scales linearly with both the number of devices $K$ and the maximum snapshot number $L^s_\text{max}$, which confirms that the proposed solution remains computationally efficient as the network size grows. 

\begin{figure*}[!t]
\setlength{\abovecaptionskip}{0.cm}
\centering
\subfigure[Number of training samples $D_k$.]{\includegraphics[scale=0.37]{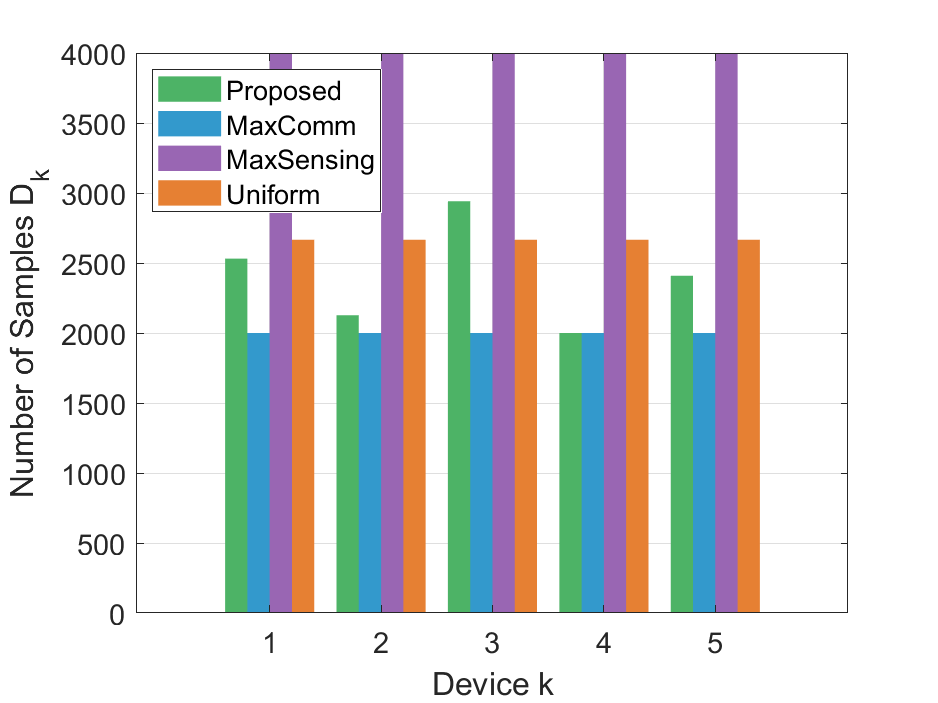}%
\label{fig_samples}}
\hfil
\subfigure[Sensing SNR $\gamma^s_k$.]{\includegraphics[scale=0.37]{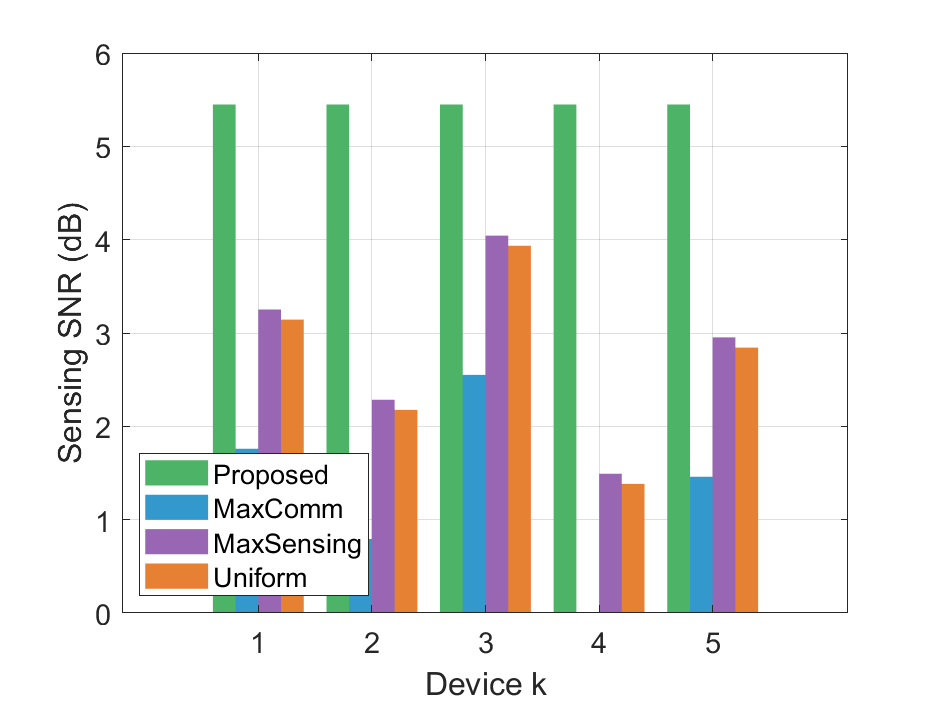}%
\label{fig_snr}}
\hfil
\subfigure[Upload success rate $q_k$.]{\includegraphics[scale=0.37]{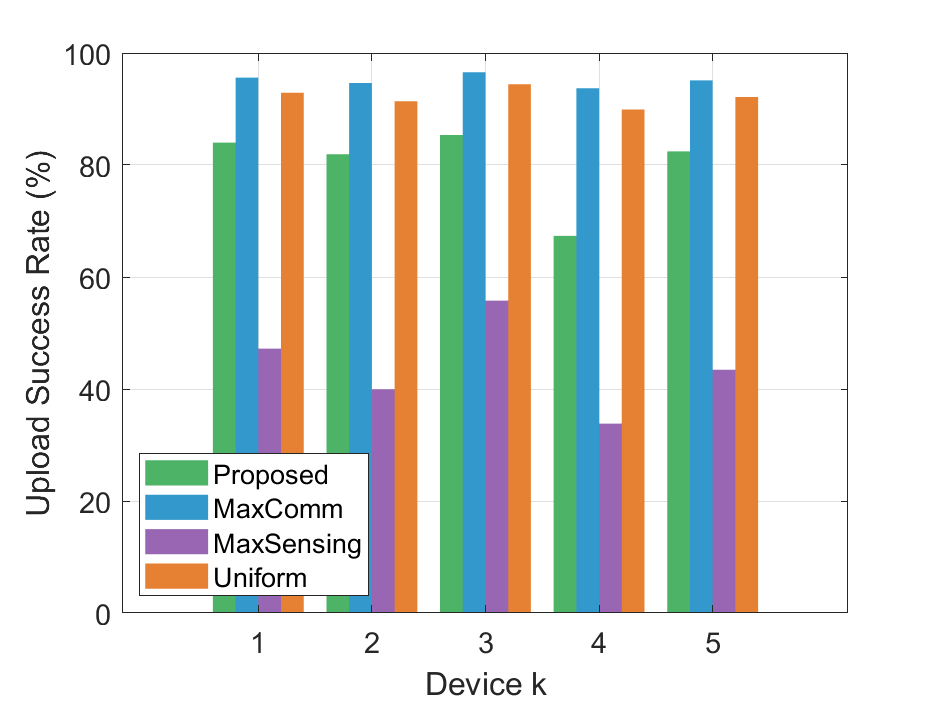}%
\label{fig_upload_rat}}
\caption{Per-device resource allocation comparison under different strategies.}
\label{fig_res_alloc}
\end{figure*}

\begin{algorithm}
\caption{Two-Layer Resource Allocation for ISAC-FL}\label{alg:alg1}
\footnotesize
\begin{algorithmic}
\STATE \textbf{Input:} System parameters $p_{\min}^s, p_{\max}^s, p_{\min}^c, p_{\max}^c, L_{\max}^s, T^s, \hat{T}^c_k, E_G$, tolerance $\varepsilon$, $\{A_k, \tau_k, E_k, \kappa_k\}_{k=1}^K$.
\STATE \textbf{Output:} Optimal resource allocation $(\mathbf{L}^{s*}, \mathbf{p}^{s*}, \mathbf{p}^{c*})$.

\STATE \textbf{Initialize:} 
\STATE $\gamma_{\text{lb}} \leftarrow p_{\min}^s \cdot \min_k\{A_k\}$, $\gamma_{\text{ub}} \leftarrow \min_k\{A_k \cdot p_{\max}^s \cdot L_{\max}^s\}$;
\STATE $\varphi \leftarrow (1 + \sqrt{5})/2$, \quad $a \leftarrow \gamma_{\text{lb}}$, \quad $b \leftarrow \gamma_{\text{ub}}$;
\WHILE{(b-a) $\geq$ $\varepsilon$:}
\STATE $\gamma_1 \leftarrow b - (b-a)/\varphi$, \quad $\gamma_2 \leftarrow a + (b-a)/\varphi$.
\STATE For each $k$, enumerate $L_k^s$ and apply \eqref{eq:optimal_ps} \eqref{eq:optimal_pc} with $\gamma_1$.
\STATE Obtain $W_1 \leftarrow \sum_k \max_{L_k^s} D_k q_k$.
\STATE For each $k$, enumerate $L_k^s$ and apply \eqref{eq:optimal_ps} \eqref{eq:optimal_pc} with $\gamma_2$.
\STATE Obtain $W_2 \leftarrow \sum_k \max_{L_k^s} D_k q_k$.
\STATE \textbf{if} $\gamma_1 \cdot W_1 < \gamma_2 \cdot W_2$: $a \leftarrow \gamma_1$; \textbf{else}: $b \leftarrow \gamma_2$.
\ENDWHILE
\STATE $\gamma^* \leftarrow (a + b) / 2$
\STATE Run \textbf{Algorithm \ref{alg:alg2}} for each $k$ with given $\gamma^*$ and obtain $(L_k^{s*}, p_k^{s*}, p_k^{c*})$.

\STATE \textbf{Return} $(\mathbf{L}^{s*}, \mathbf{p}^{s*}, \mathbf{p}^{c*})$.
\end{algorithmic}
\end{algorithm}

\begin{algorithm}
\caption{Subproblem Solution for Given $\gamma_t$}\label{alg:alg2}
\footnotesize
\begin{algorithmic}
\STATE \textbf{Input:} $\gamma_t$, device parameters $A_k, \tau_k, E_k, \kappa_k$.
\STATE \textbf{Output:} $(L_k^{s*}, p_k^{s*}, p_k^{c*})$.

\STATE \textbf{Initialize:} $W_k \leftarrow 0$.

\FOR{$L_k^s = 1$ to $L_{\max}^s$}
    \STATE Compute $D_k \leftarrow T^s / (L_k^s \cdot \tau_k)$.
    \STATE Compute $p_k^{s*}$ via \eqref{eq:optimal_ps}; \textbf{skip if} $p_k^{s*} > p_{\max}^s$.
    \STATE Compute $E_{\text{rem}} \leftarrow E_k - p_k^{s*} T^s - E_G \kappa_k D_k$; \textbf{skip if} $E_{\text{rem}} \leq 0$.
    \STATE Compute $p_k^{c*}$ via \eqref{eq:optimal_pc}.
    \STATE Update $(L_k^{s*}, p_k^{s*}, p_k^{c*})$ if $D_k q_k > W_k$.
\ENDFOR
\STATE \textbf{Return} $(L_k^{s*}, p_k^{s*}, p_k^{c*})$.
\end{algorithmic}
\end{algorithm}

\begin{table}[!ht]
\caption{Simulation Parameters\label{tab:parameters}}
\centering
\scalebox{0.9}{
\begin{tabular}{cc}
\toprule
Parameters & Values\\
\midrule
Sensing duration $T^s$ & $20$ seconds\\

Delay constraint for each round $T_{th}$ & $0.3$ seconds \\

Length of one snapshot at device $k$, $\tau_k$ & $0.0001$ seconds \\

Distance from device $k$ to the target, $d_k$ & $[10,30]$ m\\

Distance from device $k$ to the server, $d_{k,s}$ & $[20, 35]$ m\\

Antenna gain of device $k$, $G_k$ & $10$ dBi\\



Energy budget at device $k$, $E_k$ & $50$ J \\

Energy consumption per sample $\kappa_k$ & $10^{-4}$ J \\


Large-scale pathloss factor $\alpha_p$ & 4 \\

Allocated bandwidth for device $k$, $B_k$ & $1$ MHz \\

Maximal communication transmit power $p^c_\text{max}$ & 2.0 W\\

Maximal sensing transmit power $p^s_\text{max}$ & 2.0 W\\
\bottomrule
\end{tabular}}
\end{table}

\section{Numerical and Real Dataset Results}
In this section, we conduct numerical simulations to validate our
theoretical analysis and evaluate the performance of the proposed algorithms. We consider a circular area with radius $r=100$ m and $K=5$ randomly distributed ISAC devices, which sense the environment to collect training samples and participate in the federated learning framework. One BS is located at the center as the server. The other simulation parameters are summarized in Table \ref{tab:parameters}.
The number of global iterations and local iterations are $E_G=100$ and $E_L=5$.
Each point in the figures is obtained by averaging over 10 simulation runs.

\begin{figure}[!t]
\setlength{\abovecaptionskip}{0.cm}
\centering
\subfigure[Original.]{\includegraphics[width=0.75in]{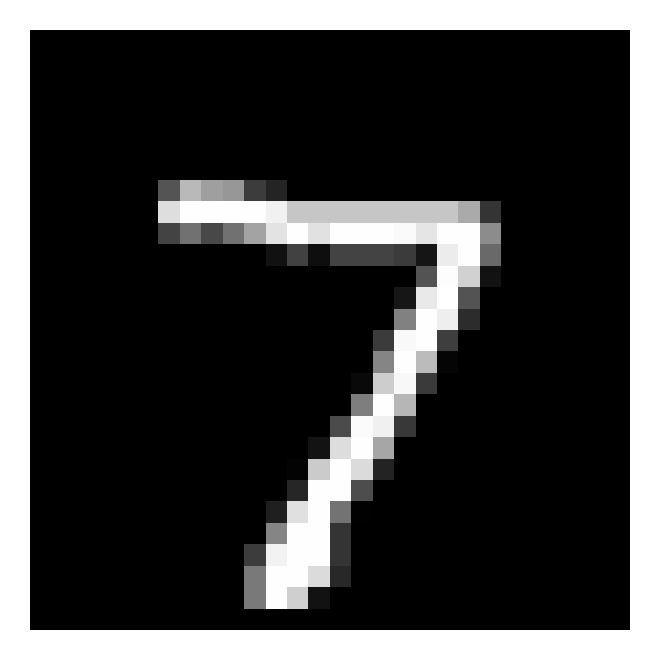}%
\label{fig_noisy_sample_original}}
\subfigure[$5$ dB.]{\includegraphics[width=0.75in]{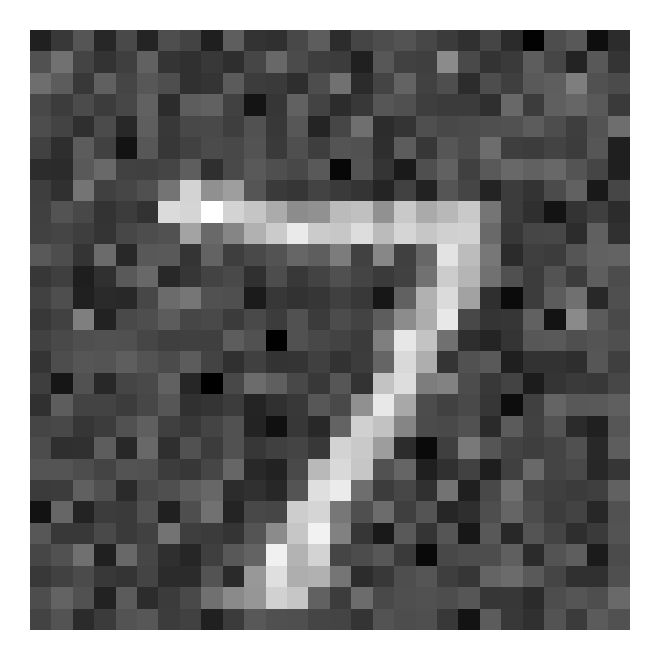}%
\label{fig_noisy_sample_5dB}}
\subfigure[$0$ dB.]{\includegraphics[width=0.75in]{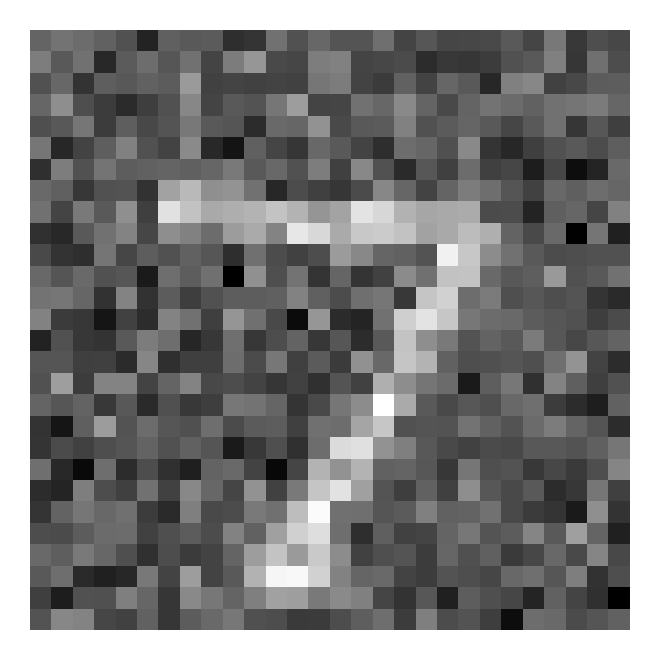}%
\label{fig_noisy_sample_0dB}}
\subfigure[$-5$ dB.]{\includegraphics[width=0.75in]{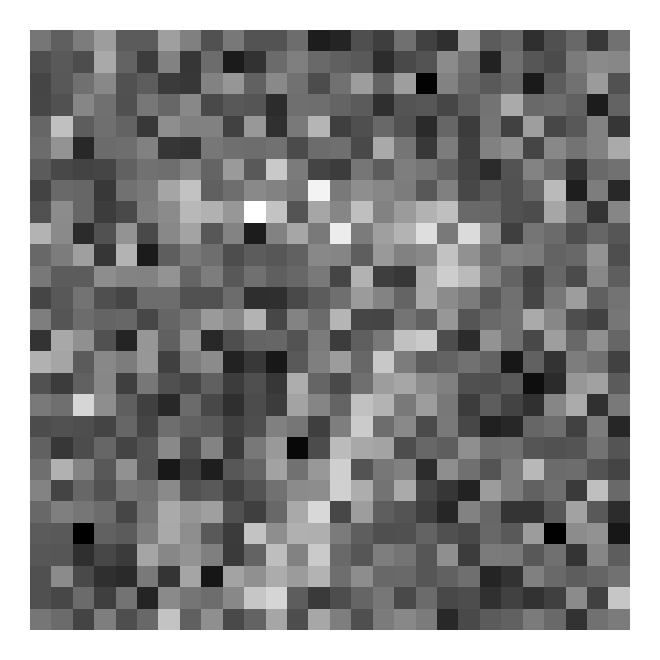}%
\label{fig_noisy_sample_neg5dB}}
\caption{Examples of MNIST digit images with AWGN noise added at different sensing SNR levels.}
\label{fig_noisy_samples}
\vspace{-12pt}
\end{figure}

\begin{figure*}[!ht]
\setlength{\abovecaptionskip}{0.cm}
\centering
\subfigure[Training loss versus upload rounds.]{\includegraphics[scale=0.38]{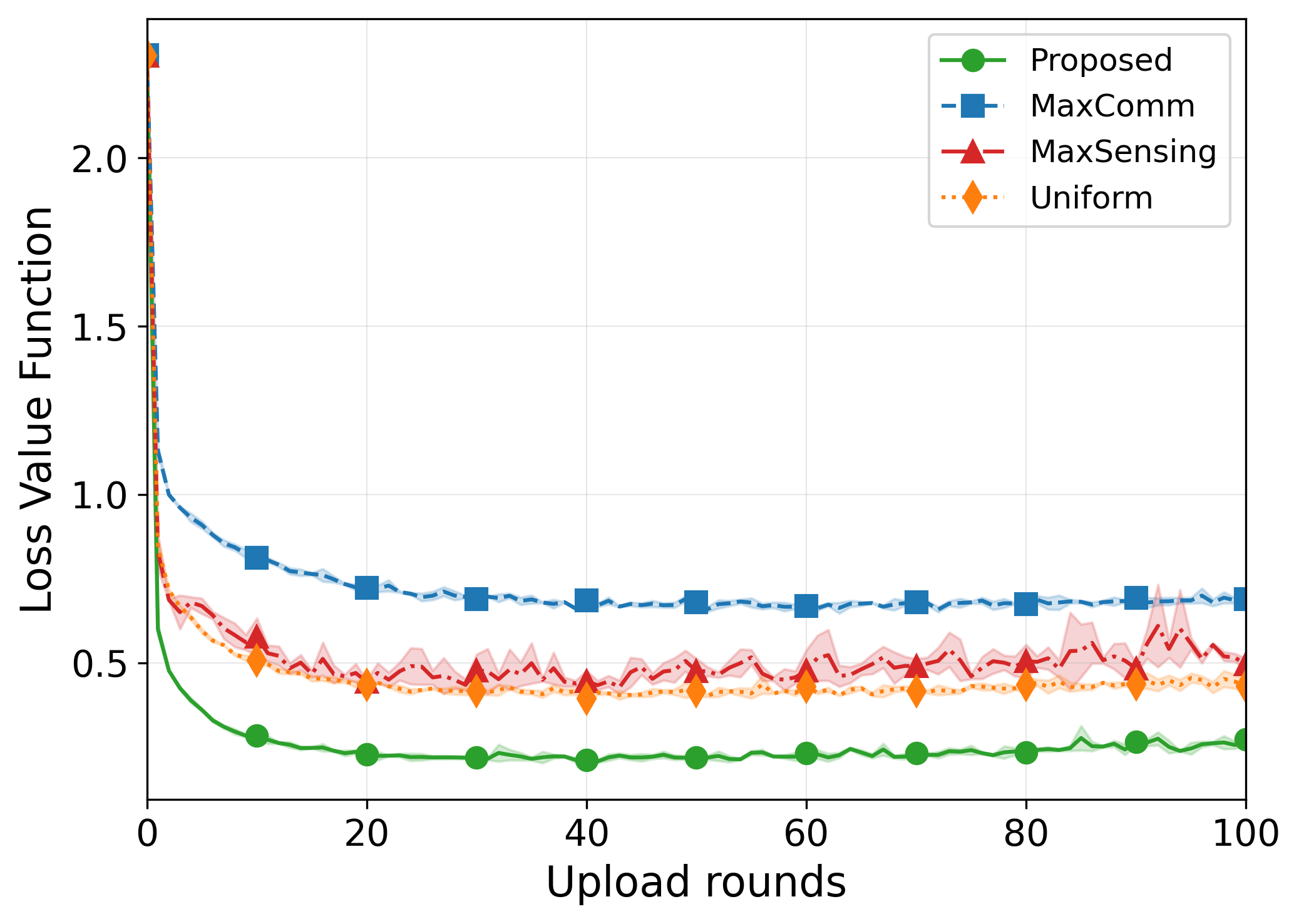}%
\label{fig_loss_iid}}
\hfil
\subfigure[Test accuracy versus upload rounds.]{\includegraphics[scale=0.38]{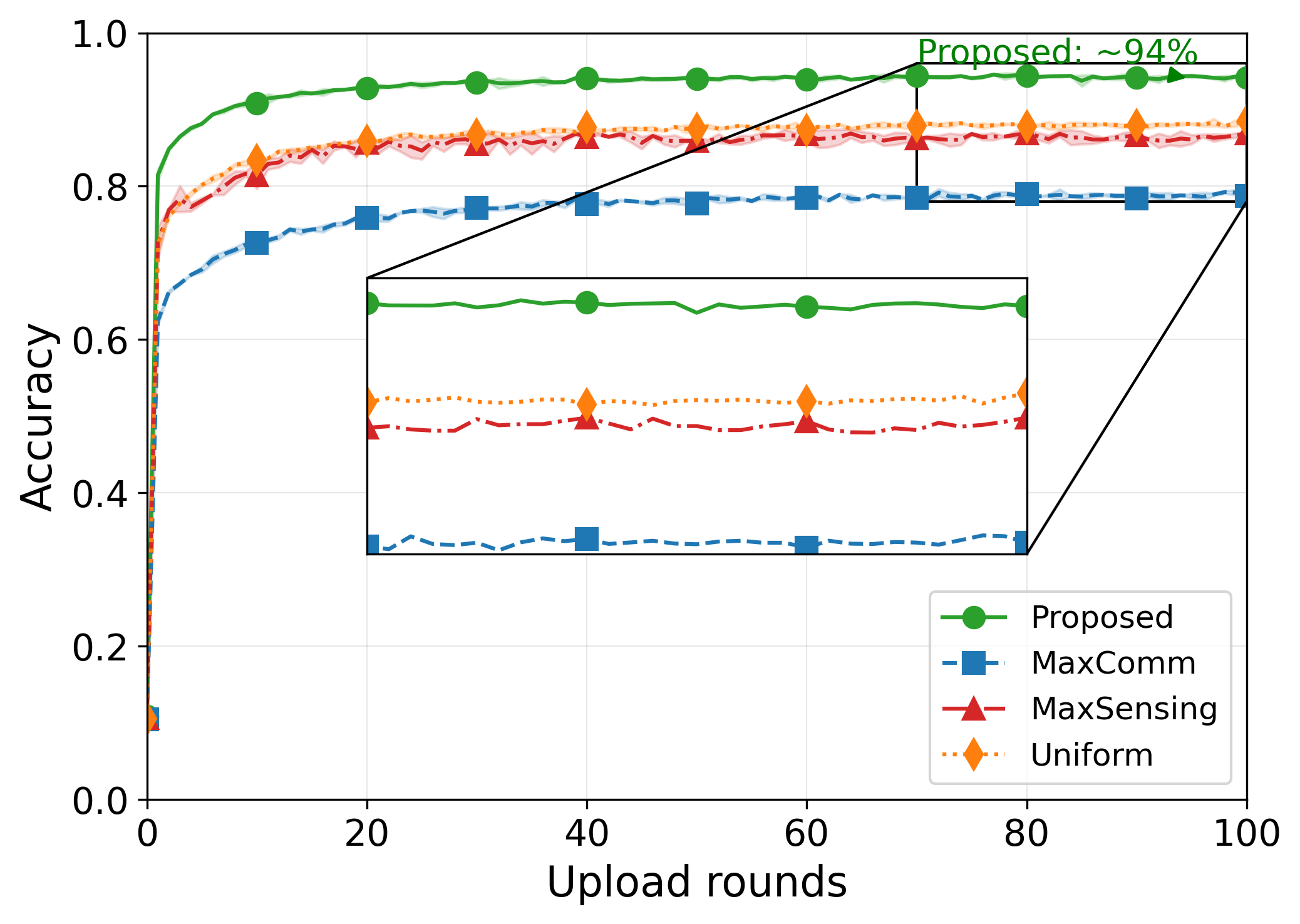}%
\label{fig_acc_iid}}
\caption{Training loss and test accuracy versus upload rounds 
under different strategies on the MNIST dataset (IID case).}
\label{fig_loss_acc_iid}
\vspace{-12pt}
\end{figure*}

\begin{figure*}[!ht]
\setlength{\abovecaptionskip}{0.cm}
\centering
\subfigure[Training loss versus upload rounds.]{\includegraphics[scale=0.38]{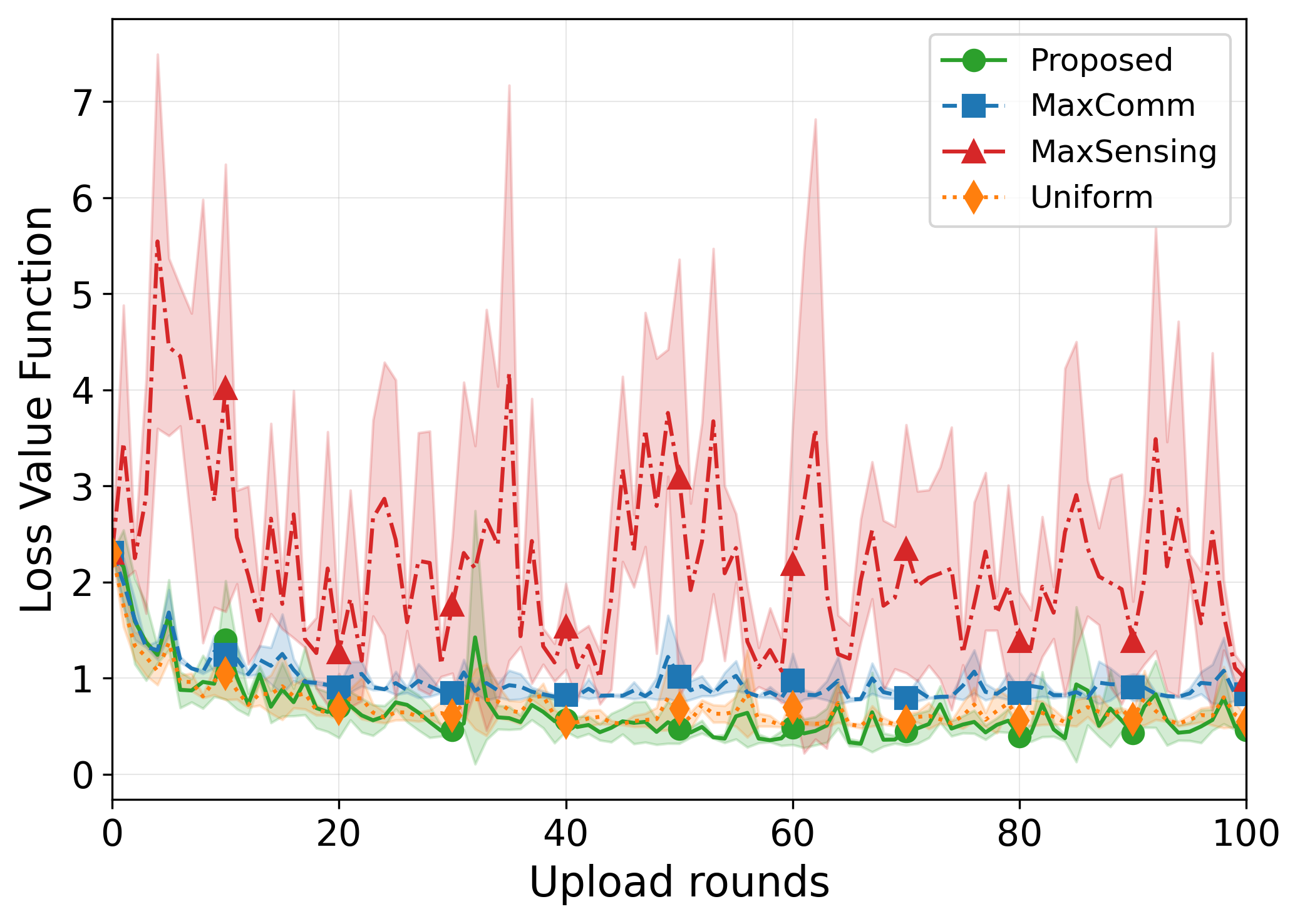}%
\label{fig_loss_niid}}
\hfil
\subfigure[Test accuracy versus upload rounds.]{\includegraphics[scale=0.38]{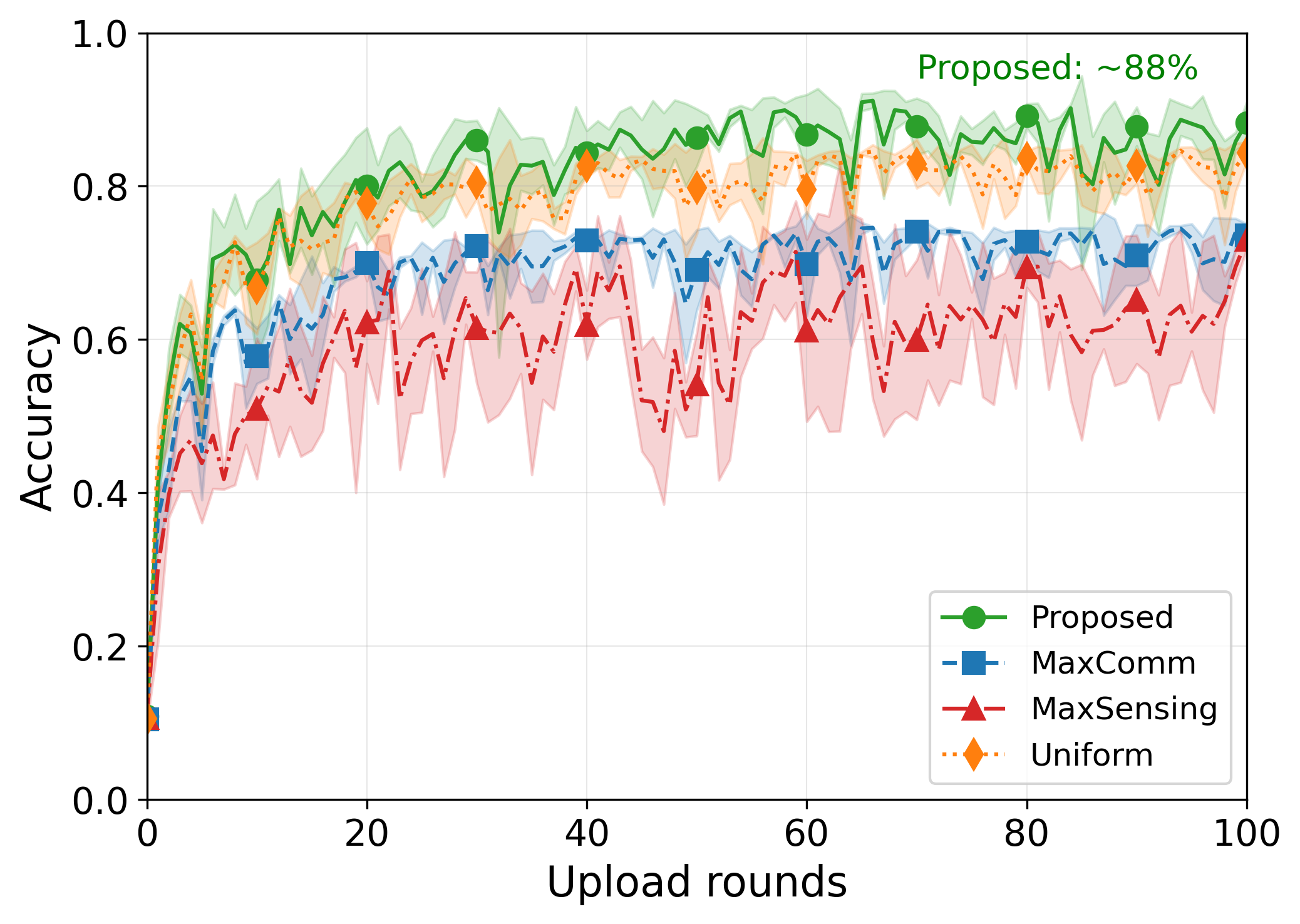}%
\label{fig_acc_niid}}
\caption{Training loss and test accuracy versus upload rounds 
under different strategies on the MNIST dataset (Non-IID case).}
\label{fig_loss_acc_niid}
\vspace{-12pt}
\end{figure*}

\begin{enumerate}
    \item \textit{Datasets}: We first adopt a comprehensive and widely used dataset MNIST \cite{lecun2002gradient}, which contains 70,000 28×28 grayscale images of handwritten digits from 0 to 9, splitting into 60,000 training and 10,000 test samples. To simulate the impact of sensing performance on training, AWGN noise is added to the local data of each device according to the corresponding sensing SNR. 
    Fig. \ref{fig_noisy_samples} shows examples of MNIST digit images corrupted by AWGN at different SNR levels. As the sensing SNR decreases, the digit features become increasingly obscured, which reflects the degraded data quality and introduces uncertainty into the local training process.
    To further evaluate the proposed ISAC-FL framework on real-world radar data, we adopt Scenario 9 of the DeepSense 6G dataset \cite{demirhan2022radar}, a large-scale multi-modal dataset that comprises coexisting multi-modal sensing and communication data, such as mmWave wireless communication, Camera, GPS data, LiDAR, and Radar, collected in realistic wireless environments. In this scenario, a vehicle equipped with a quasi-omnidirectional transmitter passes by the BS, yielding 5,964 synchronized radar measurement and beam training samples. 
    
    \item \textit{Deep Learning Model}: Devices jointly train a three-layer multilayer perceptron (MLP) with 256 neurons in each hidden layer, using LeakyReLU \cite{maas2013rectifier} as the activation function. Although the considered model is not designed to match the performance of state-of-the-art architectures, it is sufficient to clearly reveal the relative benefits of the proposed method.
    \item \textit{Baselines}: To the best of our knowledge, no existing studies have investigated the joint effect of sensing-dependent sample quantity and sample quality together with communication performance in federated learning. Consequently, no well-established baseline is available for the scenario considered in this work. For fair comparison, we consider the following baseline schemes to verify the gains brought by the joint optimization:
    \begin{itemize}
        \item \textbf{MaxSensing}: Each device uses the maximal sensing transmit power $p^s_k=p^s_\text{max}$ for collecting data.
        \item \textbf{MaxComm}: Each device uses the maximal communication transmit power $p^c_k=p^c_\text{max}$ for uploading.
        \item \textbf{Uniform}: Each device equally splits the available energy between sensing and communication, i.e., $p^s_kT^s=0.5(E_k-E_k^{train}).$
    \end{itemize}
\end{enumerate}

Fig. \ref{fig_res_alloc} illustrates the per-device resource allocation under the four strategies in terms of dataset size $D_k$, sensing SNR $\gamma^s_k$, and upload success rate $q_k$. MaxSensing strategy has the largest number of samples but at the cost of low upload reliability, while MaxComm achieves the highest upload success rate of approximately $95\%$ but yields the lowest sensing SNR close to $0$ dB. The proposed strategy maintains the highest sensing SNR at approximately $5.5$ dB across all devices while keeping $q_k$ at a moderate level of $68\%$ to $86\%$, with dataset sizes that vary per device according to their individual channel conditions and energy budgets. These results confirm that the proposed method achieves a balanced allocation that jointly benefits data quality, dataset size, and upload reliability.

Fig. \ref{fig_loss_acc_iid} shows the variation of training loss and test accuracy under IID data distribution. The proposed strategy achieves the highest test accuracy of approximately $94\%$ and outperforms all the baselines. This gain stems from the joint optimization of sensing and communication resources, which simultaneously ensures high data quality and reliable aggregation. MaxComm strategy yields the worst performance with a final accuracy of around $79\%$, indicating that prioritizing communication power at the expense of sensing quality severely degrades the training data, which cannot be compensated by high upload reliability. MaxSensing and Uniform strategies achieve comparable performance around $87\%$, the former exhibits noticeably larger variance in both loss and accuracy throughout training, reflecting the instability introduced by the unreliable aggregation. In Fig. \ref{fig_loss_acc_niid}, we also present the results under Non-IID data distribution, where all strategies experience a notable performance degradation compared to the IID case due to increased difficulty of learning when local data distributions are heterogeneous. The proposed strategy remains the optimal, achieving a final accuracy of approximately $87\%$, which demonstrates that the joint optimization is robust to data heterogeneity. Notably, MaxComm is no longer the worst in this case, since each device covers only a portion of the global distribution and reliable aggregation becomes crucial. Owing to its high upload success probability, MaxComm can still aggregate diverse local updates and achieve stable accuracy. By contrast, MaxSensing strategy emerges as the worst one, as missing updates from even a few devices introduce significant bias into the global model, and the low communication power leads to frequent upload failures that disrupt the aggregation of complementary local information. Uniform allocation achieves performance close to the proposed strategy in this setting, which suggests that a balanced resource allocation provides a reasonable heuristic when data is heterogeneous.

Fig. \ref{fig_energy_scan} illustrates the impact of communication energy ratio on the learning performance. As the communication energy ratio increases from $10\%$ to $40\%$, both the accuracy of IID and Non-IID improve along with the rapid increase in the probability of upload success. Notably, the Non-IID accuracy benefits more significantly from this improvement, since each device's local data is not representative of the global distribution, making successful aggregation of complementary updates from diverse devices particularly critical. However, both accuracy curves decline when the upload success probability approaches saturation. This indicates that once upload reliability is sufficiently high, further increasing communication power provides only limited additional benefit, while the resulting reduction in sensing power causes the sensing SNR to drop below $0$ dB. The degraded sensing quality introduces excessive noise into the local training data, and this adverse effect on data quality can no longer be compensated by the marginal improvement in upload reliability.
These results confirm the existence of an optimal energy allocation, which highlights the fundamental trade-off in the proposed ISAC-FL framework.

Fig. \ref{fig_scalability} evaluates the scalability of the four strategies as the number of devices $K$ increases from $3$ to $15$. All strategies benefit from a larger $K$ since more devices collectively contribute a greater volume of sensing data. The proposed strategy consistently achieves the highest accuracy across all values of $K$ while maintaining a substantial performance advantage over the baselines, thereby confirming the effectiveness of the joint optimization as the network size increases. Among the baselines, MaxSensing strategy exhibits the fastest performance improvement with increasing $K$. This is because MaxSensing strategy allocates more power to sensing, resulting in low upload success probability for each individual device; however, the increased number of participating devices statistically compensates for the per-device upload failures, partially recovering the aggregation quality. By contrast, MaxComm strategy shows a more moderate improvement, as its bottleneck lies in data quality rather than upload reliability. These observations confirm that the proposed strategy achieves superior scalability by jointly optimizing both sensing and communication resources.

\begin{figure}[!t]
\setlength{\abovecaptionskip}{-0.1cm}
\centering
\includegraphics[scale=0.38]{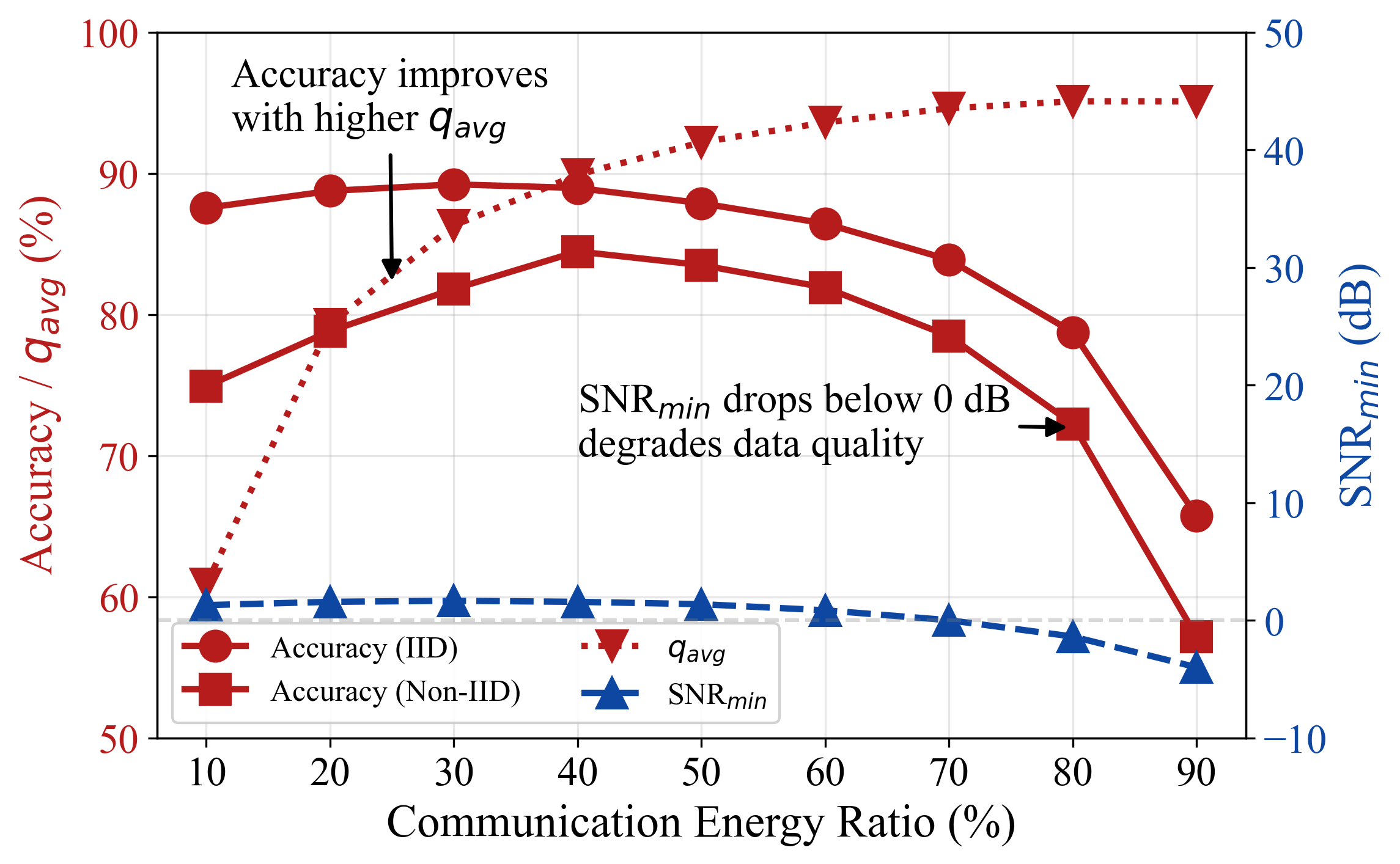}
\caption{The ISAC trade-off: test accuracy, average upload success probability $q_{\text{avg}}=\frac{1}{K}\sum_k q_k$ and minimum sensing SNR $\gamma^s_{\min}=\text{min}_k \gamma^s_k$ versus communication energy ratio under IID and Non-IID data distributions.}
\label{fig_energy_scan}
\end{figure}

\begin{figure}[!t]
\setlength{\abovecaptionskip}{-0.1cm}
\centering
\includegraphics[scale=0.36]{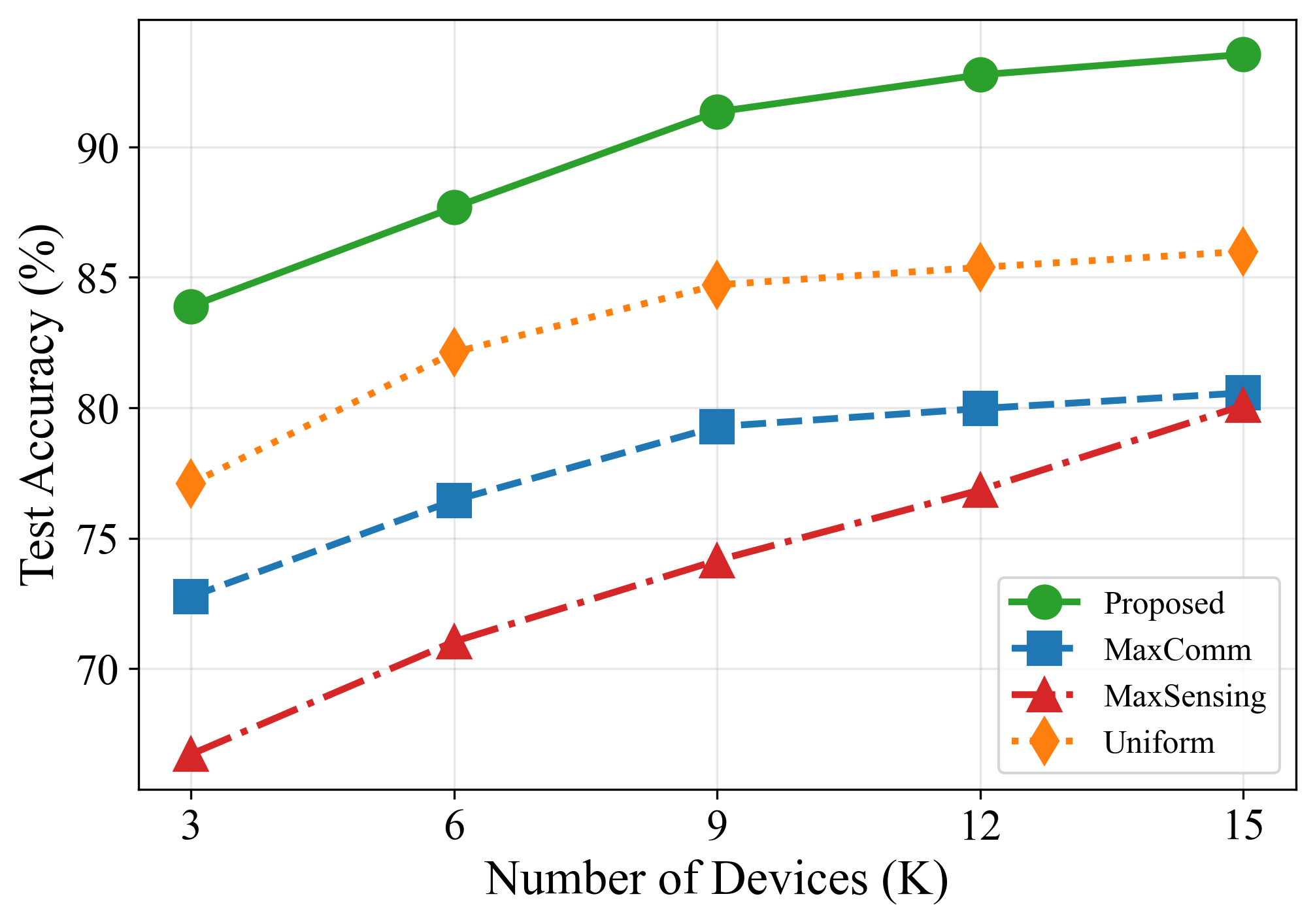}
\caption{Test accuracy versus number of devices $K$ under four resource 
allocation strategies.}
\label{fig_scalability}
\end{figure}

\begin{figure*}[!ht]
\setlength{\abovecaptionskip}{0.cm}
\centering
\subfigure[Training loss versus upload rounds.]{\includegraphics[scale=0.38]{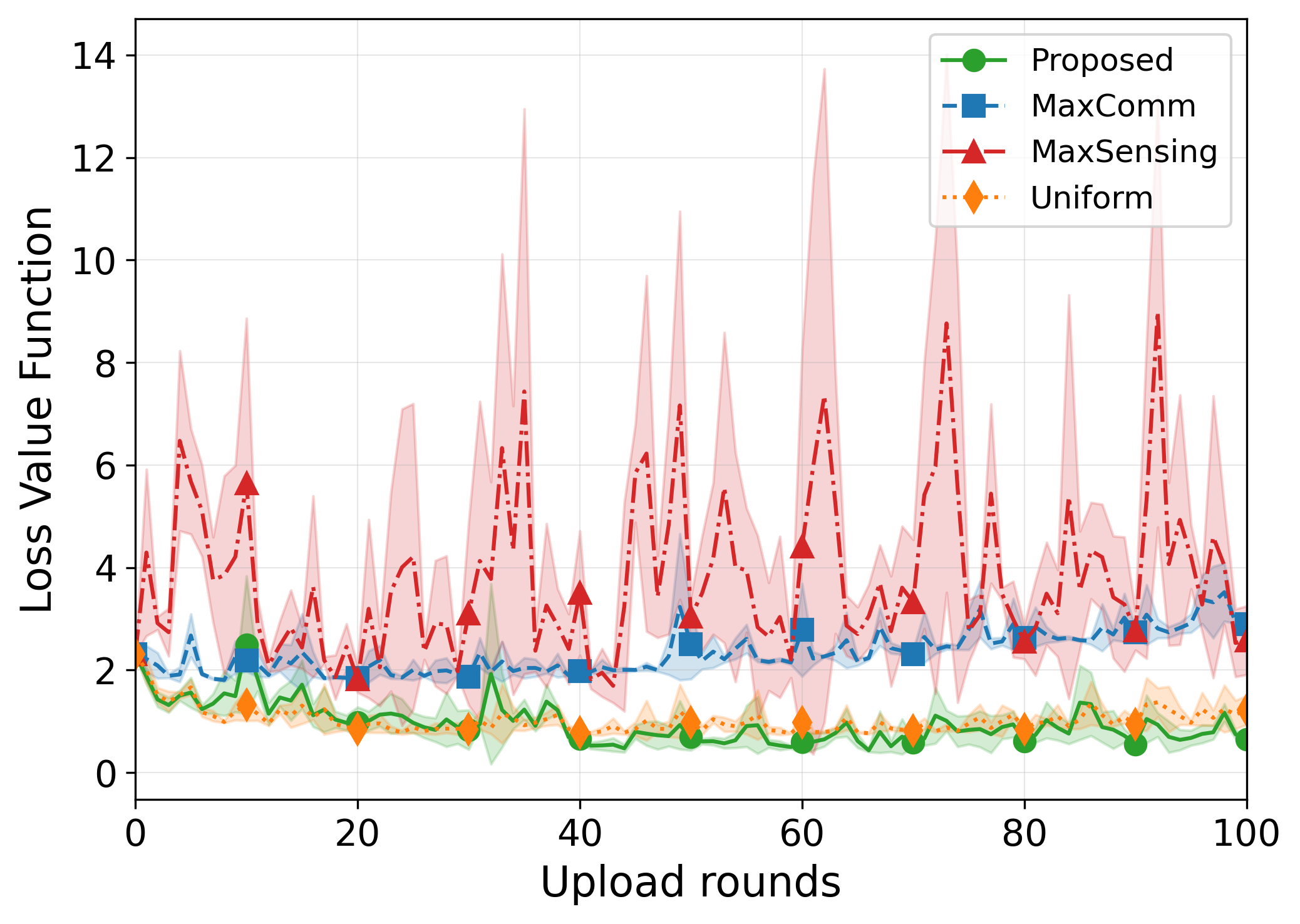}%
\label{fig_loss_strag}}
\hfil
\subfigure[Test accuracy versus upload rounds.]{\includegraphics[scale=0.38]{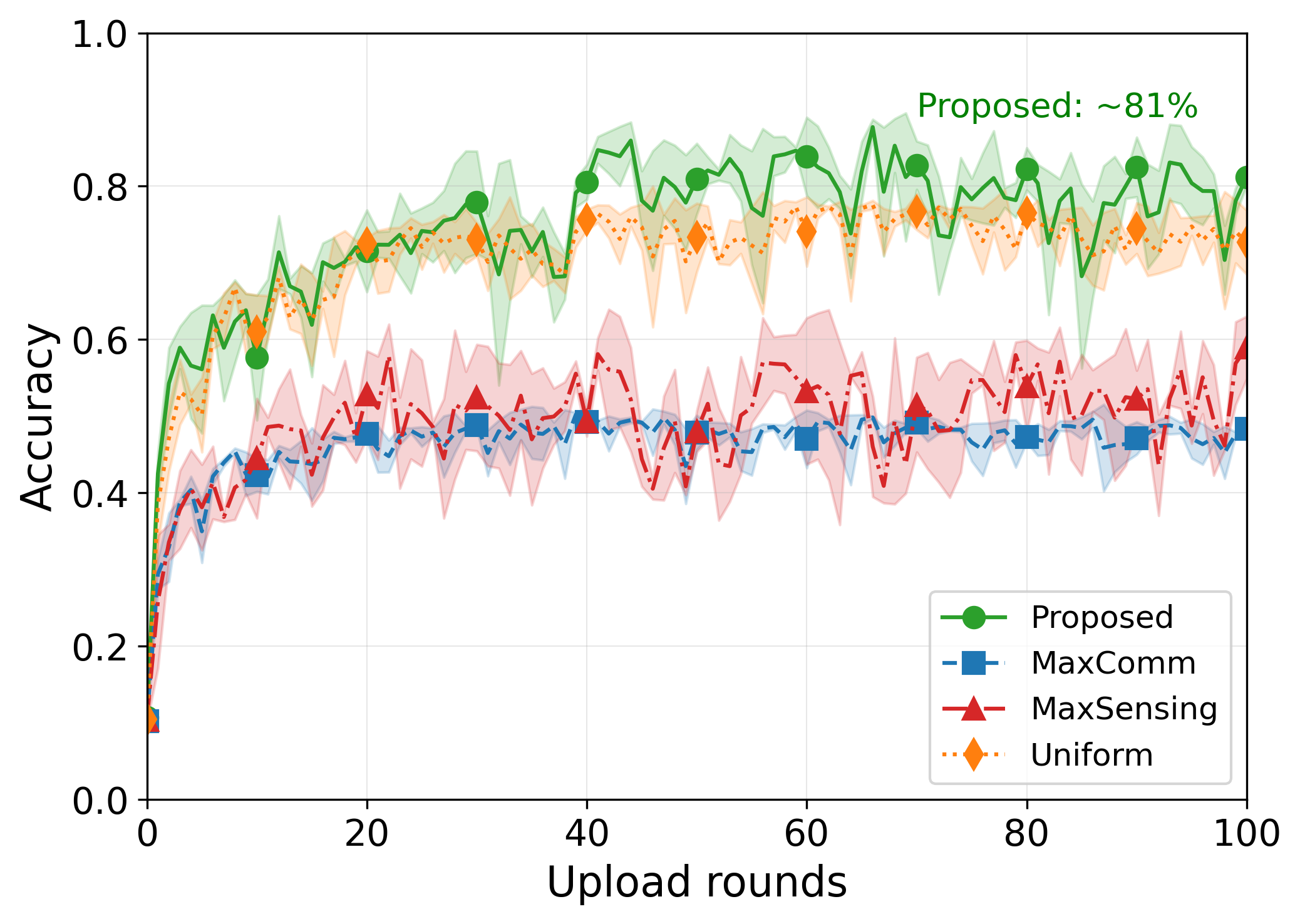}%
\label{fig_acc_strag}}
\caption{Training loss and test accuracy versus upload rounds 
under different strategies on the MNIST dataset 
(Non-IID case with one straggler device).}
\label{fig_loss_acc_strag}
\vspace{-12pt}
\end{figure*}

Fig. \ref{fig_loss_acc_strag} presents the results under the Non-IID setting in the presence of a straggler device with a significantly reduced energy budget, which is a practical challenge in FL systems \cite{reisizadeh2022straggler, pei2024review}. Compared to the homogeneous case, all strategies experience further performance degradation. The proposed strategy remains optimal, and demonstrates greater robustness to the straggler than the baselines. 
Notably, MaxComm strategy suffers a more severe degradation compared to the homogeneous case, with accuracy dropping to around $50\%$. 
This is because forcing the straggler to allocate the maximum power to communication leaves it with a very limited local dataset of poor quality, which prevents the global model from capturing an important portion of the overall data distribution and consequently degrades the aggregation quality.

Moreover, we extend the DL-based beam prediction framework of \cite{demirhan2022radar} to a FL setting. Following the preprocessing pipeline in \cite{demirhan2022radar}, we apply the radar cube transformation to the raw radar measurements, and adopt the same LeNet-based CNN architecture as the local model for each device. The training performance is evaluated using the cumulative top-$k$ accuracy for $k=1,2,3$, which measures the probability that the optimal beam index is contained within the top-$k$ predicted candidates. Fig. \ref{fig_topk} presents the comparison of training performance under different strategies on the DeepSense 6G dataset. The proposed strategy consistently achieves the highest accuracy from top-$1$ through top-$3$, with values of $43.6\%$, $63.4\%$ and $76.2\%$ respectively.
These results demonstrate that the joint optimization of dataset size, sensing SNR, and upload success probability yields a more discriminative model, enabling more accurate beam prediction. The consistent superiority of the proposed strategy validates the effectiveness of the joint optimization framework in real-world radar-aided communication scenarios.

\begin{figure}[!t]
\setlength{\abovecaptionskip}{-0.1cm}
\centering
\includegraphics[scale=0.35]{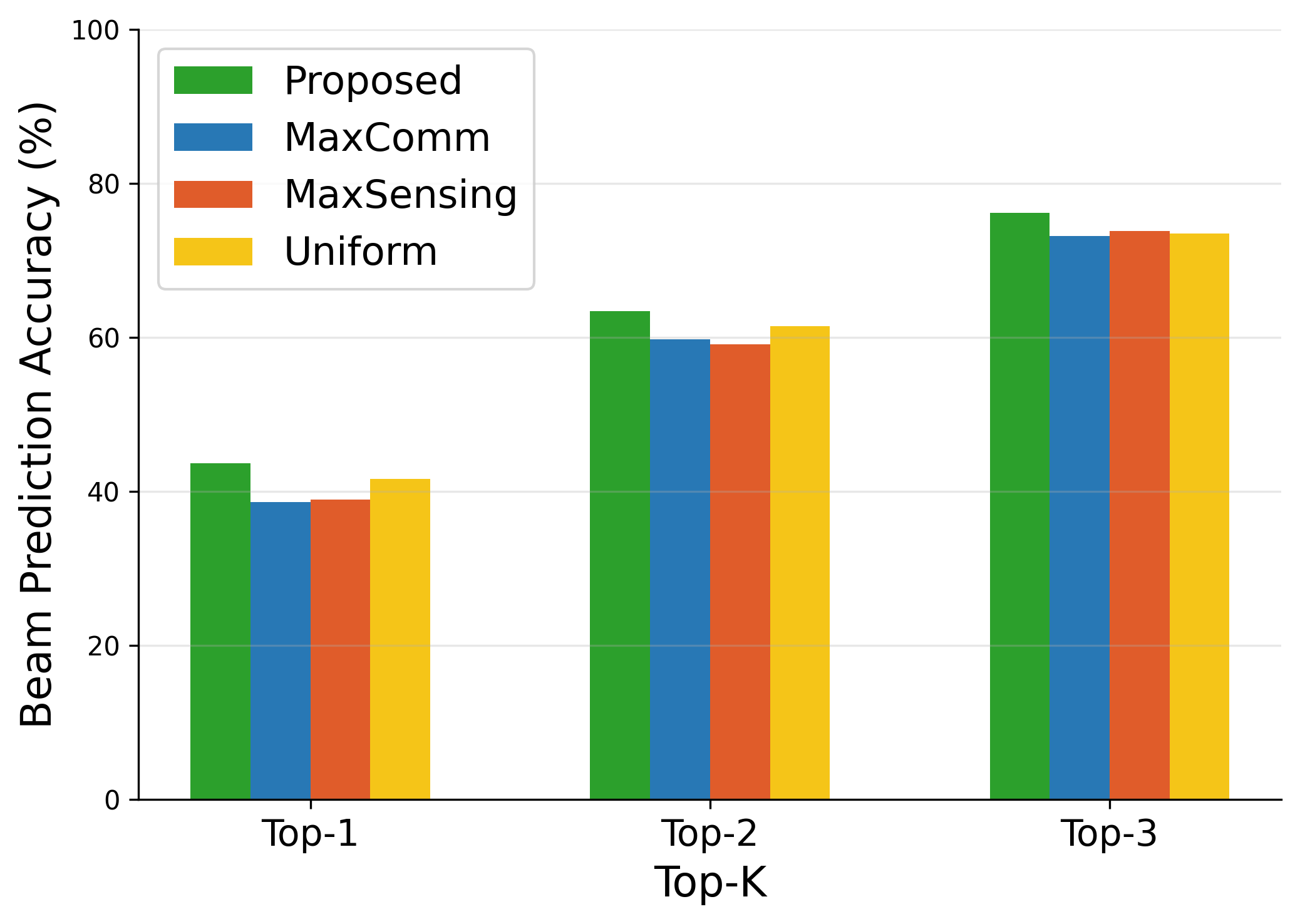}
\caption{Top-$k$ ($k=1,2,3$) beam prediction accuracy under four different strategies evaluated on the DeepSense 6G Scenario 9 dataset.}
\label{fig_topk}
\vspace{-12pt}
\end{figure}

\section{Conclusion}
In this paper, we have investigated the resource allocation problem for federated learning in ISAC systems, where devices sense the environment to collect training data and upload model updates to a server. 
We analytically show that the FL optimality gap depends explicitly on the sensing SNR, the local dataset sizes, and the upload success probabilities. Building on this result, we derive a closed-form upper bound on convergence performance that characterizes the combined effect of these factors.
Based on this bound, the original intractable optimization problem was reformulated into a tractable resource allocation problem, and an efficient algorithm has been developed with a complexity that scales linearly with the number of devices and the maximum snapshot number. Numerical results on both the MNIST dataset and the real-world DeepSense 6G radar dataset demonstrated that the proposed joint optimization strategy consistently outperforms baseline schemes under both IID and Non-IID data distributions, and confirmed the existence of an optimal energy allocation that balances sensing quality against upload reliability. 

\section*{Appendix A: \textsc{Notations}}
For simplicity, some shorthand notations are defined as
follows.
$\mathbf{w}^k_t$: the local model of device $k$ at round $t$.

$\tilde{g}_k(\mathbf{w}^k_t;\xi)$: the stochastic gradient computed at device $k$ on a random sample $\xi$.

$g_k(\mathbf{w}^k_t)=E_{\xi}[\tilde{g_k}(\mathbf{w}^k_t;\xi)]=\nabla F_k(\mathbf{w}^k_t) $: the expected gradient at device $k$ at round $t$.

Aggregation weight:
\begin{small}
\begin{equation}
    E[\rho^{(t)}_k]=\bar{\rho}_k \approx \frac{E[D_kC_k^{(t)}]}{E[W_t]}= \frac{D_k q_k}{\sum_kD_k q_k},
\end{equation}
\end{small}where we adopt the approximation that $E[D_kC_k^{(t)}/W_t]\approx E[D_kC_k^{(t)}]/E[W_t]$. The aggregation update is skipped for rounds with $W_t=0$.

Aggregated stochastic gradient is given by:
\begin{equation}
    \tilde{g}_t= \sum_{k\in S_t}\rho^{(t)}_k\tilde{g}_k(\mathbf{w}^k_t;\xi).
\end{equation}

Expected aggregated gradient is given by:
\begin{equation}
\label{eq:no1}
    g_t = E[\tilde{g}_t]\approx\nabla F(\mathbf{w}_t).
\end{equation}
Here, the aggregated stochastic gradient $E[\tilde{g}_t]$ is approximately unbiased by treating the random aggregation weights by their average behavior

Global model update is given by:
\begin{equation}
    \mathbf{w}_{t+1} = \mathbf{w}_t - \eta \tilde{g}_t.
\end{equation}

\section*{Appendix B: \textsc{Proof of Proposition 1}}
Let $\tilde{\xi}=\xi + n$ denote the sensing sample corrupted by noise, where $n \sim \mathcal{N}(0, \sigma^2_\varepsilon \mathbf{I})$ is the AWGN introduced during the sensing process.
We additionally assume that the loss function gradient is $L_x$-Lipschitz continuous with respect to the input $\tilde{\xi}$, i.e., $\parallel \nabla F_k(\mathbf{w}^t_k;\tilde{\xi})-\nabla F_k(\mathbf{w}^t_k;\xi) \parallel \leq L_x \parallel \tilde{\xi}-\xi \parallel$, which holds for smooth neural network architectures with bounded weights. 
The stochastic gradient computed on the noisy sample can be decomposed as
\begin{small}
\begin{equation}
    \begin{aligned}
        \tilde{g}_k(\mathbf{w}^t_k;\tilde{\xi})-g_k(\mathbf{w}^t_k)&=\Big[\tilde{g}_k(\mathbf{w}^t_k;\tilde{\xi})-\nabla F_k(\mathbf{w}^t_k;\xi) \Big] \\
        &+ \Big[\nabla F_k(\mathbf{w}^t_k;\xi)-g_k(\mathbf{w}^t_k) \Big],
    \end{aligned}
\end{equation}
\end{small}where the first term on the right-hand side is the bias induced by the sensing noise and the second term denotes the SGD stochastic variance.
Applying the triangle inequality, the gradient noise variance satisfies
\begin{small}
\begin{equation}
    \begin{aligned}
        E[\parallel \tilde{g}_k-g_k \parallel^2] &\leq 2E[\parallel \tilde{g}_k(\mathbf{w}^t_k;\tilde{\xi})-\nabla F_k(\mathbf{w}^t_k;\xi) \parallel^2] \\
        &+ 2E[\parallel \nabla F_k(\mathbf{w}^t_k;\xi)-g_k(\mathbf{w}^t_k) \parallel^2].
    \end{aligned}
\end{equation}
\end{small}For the first term, invoking the $L_x$-Lipschitz continuity of the gradient with respect to the input yields
\begin{small}
\begin{equation}
    \begin{aligned}
        E[\parallel \tilde{g}_k(\mathbf{w}^t_k;\tilde{\xi})-\nabla F_k(\mathbf{w}^t_k;\xi) \parallel^2] &\leq \frac{L_x^2E[\parallel n \parallel^2]}{D_k} \\
        &= \frac{L_x^2E[\parallel \xi \parallel^2]}{D_k\gamma^s_k},
    \end{aligned}
\end{equation}
\end{small}where we used the SNR definition $\gamma^s_k=E[\parallel \xi \parallel^2]/E[\parallel n \parallel^2]$ and $\tilde{g}_k=\frac{1}{D_k}\sum_{i=1}^{D_k}\nabla F_k(\mathbf{w}^t_k;\tilde{\xi}_i)$. For the second term, the standard bounded variance assumption for mini-batch SGD over $D_k$ samples gives \cite{li2019convergence}
\begin{small}
\begin{equation}
    \begin{aligned}
        E[\parallel \nabla F_k(\mathbf{w}^t_k;\xi)-g_k(\mathbf{w}^t_k) \parallel^2] \leq \frac{c_0}{D_k},
    \end{aligned}
\end{equation}
\end{small}
Combining both terms, we obtain
\begin{small}
\begin{equation}
    \begin{aligned}
        E[\parallel \tilde{g}_k-g_k \parallel^2] \leq \frac{2L_x^2 E[\parallel \xi \parallel^2]}{D_k\gamma^s_k} + \frac{2c_0}{D_k}.
    \end{aligned}
\end{equation}
\end{small}
Since the system constraints \eqref{eq:cost1} and \eqref{eq:cost3} impose $\gamma^s_k \leq \gamma^s_\text{max}$, where $\gamma^s_\text{max}$ denotes the maximum achievable sensing SNR, the second term can be further bounded as $\frac{2c_0}{D_k} \leq \frac{2c_0 \gamma^s_\text{max}}{D_k \gamma^s_k}$. Letting $\sigma^2_0=2 L_x^2 E[\parallel \xi \parallel^2] + 2c_0 \gamma^s_\text{max} $, we have
\begin{small}
\begin{equation}
    E[\parallel \tilde{g}_k-g_k \parallel^2] \leq \frac{\sigma^2_0}{D_k \gamma^s_k} \triangleq \frac{\sigma^2_k(\gamma^s_k)}{D_k},
\end{equation}
\end{small}where $\sigma^2_k(\gamma^s_k)=\sigma^2_0 / \gamma^s_k$. This result reveals that the gradient noise upper bound is inversely proportional to both the sensing SNR $\gamma^s_k$ and the local dataset size $D_k$.

\section*{Appendix C: \textsc{Proof of Lemma \ref{lm:gradient_noise}}}
From the definition and Assumption 3, it follows that
\begin{small}
\begin{equation}
\label{eq:lm_eq1}
    \begin{aligned}
        E[\parallel \varepsilon_t \parallel^2] &=E[\parallel \tilde{g}_t - g_t \parallel^2]\\
        &= E\Big[\parallel \sum^K_{k=1}\frac{D_kC^{(t)}_k}{W_t}( \tilde{g}_k(\mathbf{w}^k_t;\xi) - \nabla F_k(\mathbf{w}^k_t)) \parallel^2 \Big]. \\
    \end{aligned}
\end{equation}
\end{small}
Since the noise across different devices is independent and the noise-induced gradient bias has zero mean conditioned on $\mathbf{w}^k_t$, when $k \neq j$, we have
\begin{small}
\begin{equation}
\label{eq:lm_eq2}
\begin{aligned}
    E[(\tilde{g}_k-\nabla F_k)^T (\tilde{g}_j-\nabla F_j)]=E[\tilde{g}_k-\nabla F_k]^TE[\tilde{g}_j-\nabla F_j]
    =0,
\end{aligned}
\end{equation}
\end{small}
Hence, \eqref{eq:lm_eq1} becomes
\begin{small}
\begin{equation}
\label{eq:lm_eq3}
    \begin{aligned}
        E[\parallel \varepsilon_t \parallel^2] &= E\Big[\parallel \sum^K_{k=1}\frac{D_kC^{(t)}_k}{W_t}( \tilde{g}_k(\mathbf{w}^k_t;\xi) - \nabla F_k(\mathbf{w}^k_t) \parallel^2 \Big]\\
        &=E\Big[\sum^K_{k=1}\Big(\frac{D_kC^{(t)}_k}{W_t}\Big)^2\parallel \tilde{g}_k(\mathbf{w}^k_t;\xi) - \nabla F_k(\mathbf{w}^k_t) \parallel^2 \Big]\\
        &=E\Big[\sum^K_{k=1}\frac{D_k^2C^{(t)}_k}{W_t^2}\parallel \tilde{g}_k(\mathbf{w}^k_t;\xi) - \nabla F_k(\mathbf{w}^k_t) \parallel^2 \Big] \\
        &\overset{a}{\leq} E\Big[\sum^K_{k=1}\rho^{(t)}_k\parallel \tilde{g}_k(\mathbf{w}^k_t;\xi) - \nabla F_k(\mathbf{w}^k_t) \parallel^2 \Big]
    \end{aligned}
\end{equation}
\end{small}
where (a) stems from
\begin{small}
\begin{equation}
\label{eq:lm_eq4}
    \Big(\frac{D_kC^{(t)}_k}{W_t}\Big)^2=\frac{D_k^2C^{(t)}_k}{W_t^2}\leq \frac{D_kC^{(t)}_k}{W_t}\triangleq\rho^{(t)}_k.
\end{equation}
\end{small}
Let us assume the sensing SNR satisfies $\gamma^s_k\geq\gamma^s_{\text{min}}> 0$, under Assumption 3, we have
\begin{small}
\begin{equation}
\label{eq:lm_eq5}
    E\Big[\parallel \tilde{g}_k(\mathbf{w}^k_t;\xi) - \nabla F_k(\mathbf{w}^k_t) \parallel^2 \Big]= \frac{\sigma^2_0}{D_k\gamma^s_k}\leq \frac{\sigma^2_0}{D_k\gamma^s_{\text{min}}},
\end{equation}
\end{small}
then by substituting \eqref{eq:lm_eq5} into \eqref{eq:lm_eq3}, and applying $\sum_k q_k \leq K$ as $q_k \leq 1$ always holds, we can obtain
\begin{small}
\begin{equation}
\label{eq:lm_eq6}
    \begin{aligned}
        E[\parallel \varepsilon_t \parallel^2] &\leq \frac{\sigma^2_0}{\gamma^s_{\text{min}}}\sum_{k=1}^K \frac{E\Big[\rho^{(t)}_k \Big]}{D_k}\\
        &\approx \frac{\sigma^2_0}{\gamma^s_{\text{min}}}\sum_{k=1}^K \frac{1}{D_k}\frac{D_kq_k}{\sum_k D_kq_k}\\
        &=\frac{\sigma^2_0}{\gamma^s_{\text{min}}} \frac{\sum_kq_k}{\sum_k D_kq_k} \\
        &\leq \frac{\sigma^2_0}{\gamma^s_{\text{min}}} \frac{K}{\sum_k D_kq_k},
    \end{aligned}
\end{equation}which completes the proof.
\end{small}

\section*{Appendix D: \textsc{Proof of Theorem \ref{theo:convergence}}}

From Assumption 2 and $\mathbf{w}_{t+1}=\mathbf{w}_t-\eta \tilde{g}_t$, we have
\begin{small}
\begin{equation}
\label{eq:th_eq1}
\begin{aligned}
    F(\mathbf{w}_{t+1}) &\leq F(\mathbf{w}_t) + \langle\nabla F(\mathbf{w}_t),\mathbf{w}_{t+1}-\mathbf{w}_t \rangle \\
    &+ \frac{L}{2}\parallel\mathbf{w}_{t+1}-\mathbf{w}_t\parallel^2\\
    &=F(\mathbf{w}_t)-\eta \langle \nabla F(\mathbf{w}_t), \tilde{g}_t \rangle + \frac{L\eta^2}{2}\parallel \tilde{g}_t \parallel^2.
\end{aligned}
\end{equation}
\end{small}

Under the approximation in \eqref{eq:no1}, take expectations on both sides of \eqref{eq:th_eq1} gives
\begin{small}
\begin{equation}
\label{eq:th_eq2}
\begin{aligned}
    E[F(\mathbf{w}_{t+1})] &\leq E[F(\mathbf{w}_t)]-\eta E[\langle \nabla F(\mathbf{w}_t), \tilde{g}_t \rangle]\\ 
    &+ \frac{L\eta^2}{2} E[\parallel \tilde{g}_t \parallel^2]\\
    &\overset{b}{=}E[F(\mathbf{w}_t)]-\eta E[\parallel \nabla F(\mathbf{w}_t) \parallel^2]\\
    &+ \frac{L\eta^2}{2} E[\parallel \tilde{g}_t \parallel^2].
\end{aligned}
\end{equation}
\end{small}
where $(b)$ stems from $E[\langle \nabla F(\mathbf{w}_t), \tilde{g}_t \rangle]=E[\parallel \nabla F(\mathbf{w}_t) \parallel^2]$.

Let $\tilde{g}_t=\nabla F(\mathbf{w}_t)+(\tilde{g}_t-\nabla F( \mathbf{w}_t))$, and based on the triangle inequality, it follows that
\begin{small}
\begin{equation}
\label{eq:th_eq3}
    \begin{aligned}
        E[\parallel \tilde{g}_t \parallel^2] &=  E[\parallel \nabla F(\mathbf{w}_t)+ (\tilde{g}_t - \nabla F(\mathbf{w}_t)) \parallel^2]\\
        &\leq 2E[\parallel \nabla F(\mathbf{w}_t) \parallel^2] + 2E[\parallel \tilde{g}_t - \nabla F(\mathbf{w}_t) \parallel^2].
    \end{aligned}
\end{equation}
\end{small}

Substituting \eqref{eq:th_eq3} into \eqref{eq:th_eq2}, we have
\begin{small}
\begin{equation}
\label{eq:th_eq4}
    \begin{aligned}
        E[F(\mathbf{w}_{t+1})] &\leq E[F(\mathbf{w}_t)]+(L\eta^2-\eta) E[\parallel \nabla F(\mathbf{w}_t) \parallel^2]\\
        &+L\eta^2 E[\parallel \tilde{g}_t - \nabla F(\mathbf{w}_t) \parallel^2].
    \end{aligned}
\end{equation}
\end{small}

Let $\triangle_t=\tilde{g}_t - \nabla F(\mathbf{w}_t)$, next we show that $E[\parallel \triangle_t \parallel^2]$ is bounded. From the definition, there is
\begin{small}
\begin{equation}
\label{eq:th_eq5}
    \begin{aligned}
        \triangle_t &=\sum_{k\in S_t}\rho_k^{(t)}(\tilde{g}_k-g_k) + \sum_{k\in S_t}\rho_k^{(t)} (g_k-\nabla F(\mathbf{w}_t)),
    \end{aligned}
\end{equation}
\end{small}where the first term captures the noise due to local SGD, while the second term accounts for the aggregation error. Applying triangle inequality again to \eqref{eq:th_eq5}, we have
\begin{small}
\begin{equation}
\label{eq:th_eq6}
    \begin{aligned}
        E[\parallel \triangle_t \parallel^2] &= E\parallel \sum_{k\in S_t}\rho_k^{(t)}(\tilde{g}_k-g_k) + \sum_{k\in S_t}\rho_k^{(t)} (g_k-\nabla F(\mathbf{w}_t)) \parallel^2\\
        &\leq 2E\parallel\sum_{k\in S_t}\rho_k^{(t)}(\tilde{g}_k-g_k)\parallel^2 \\
        &+ 2E\parallel \sum_{k\in S_t}\rho_k^{(t)} (g_k-\nabla F(\mathbf{w}_t)) \parallel^2,\\
    \end{aligned}
\end{equation}
\end{small}the two terms on the right-hand side of \eqref{eq:th_eq6} can be relaxed by using Young's inequality and Cauchy-Schwarz inequality respectively, given by
\begin{small}
\begin{equation}
\label{eq:th_eq7}
    \begin{aligned}
        E[\parallel \triangle_t \parallel^2] &\leq 2E\parallel\sum_{k\in S_t}\rho_k^{(t)}(\tilde{g}_k-g_k)\parallel^2\\
        &+ 2E \parallel \sum_{k\in S_t}\rho_k^{(t)} (g_k- \nabla F(\mathbf{w}_t)) \parallel^2 \\
        &\leq 2E\Big [\sum_{k\in S_t}(\rho_k^{(t)}) \parallel (\tilde{g}_k-g_k) \parallel^2 \Big]\\
        &+ 2E \Big[ \Big ( \sum_{k\in S_t}(\rho_k^{(t)})^2\Big ) \Big( \sum_{k\in S_t}\parallel g_k- \nabla F(\mathbf{w}_t) \parallel^2 \Big )\Big]
    \end{aligned}
\end{equation}
\end{small}
Based on Lemma \ref{lm:gradient_noise} and Assumption 4, we obtain
\begin{small}
\begin{equation}
\label{eq:th_eq8}
    \begin{aligned}
        E[\parallel \triangle_t \parallel^2] &\leq \frac{2\sigma^2_0}{\gamma^s_{\text{min}}}\frac{K}{\sum_k D_kq_k} \\
        &+ 2E \Big[ \Big ( \sum_{k\in S_t}(\rho_k^{(t)})^2\Big ) \Big( \sum_{k\in S_t}\parallel g_k- \nabla F(\mathbf{w}_t) \parallel^2 \Big )\Big]\\
        &\leq \frac{2\sigma^2_0}{\gamma^s_{\text{min}}}\frac{K}{\sum_k D_kq_k} + 2\zeta^2 E[\sum^{K}_{k=1}(\rho^{(t)}_k)^2]\\
    \end{aligned}
\end{equation}
\end{small}


Combining \eqref{eq:th_eq8} and \eqref{eq:th_eq4}, and applying $E[(\rho_k^{(t)}
)^2]\approx \bar{\rho}_k^2$, it follows that
\begin{small}
\begin{equation}
\label{eq:th_eq9}
    \begin{aligned}
        E[F(\mathbf{w}_{t+1})] &\leq E[F(\mathbf{w}_t)]+(L\eta^2-\eta) E[\parallel \nabla F(\mathbf{w}_t) \parallel^2]\\
        &+ 2L\eta^2 \Big (\frac{\sigma^2_0}{\gamma^s_\text{min}} \frac{ K}{\sum_kD_kq_k}+ \zeta^2\sum^K_{k=1}\bar{\rho}_k^2 \Big).
    \end{aligned}
\end{equation}
\end{small}

From Assumption 1, we know that \cite{boyd2004convex}
\vspace{-1.5mm}
\begin{small}
\begin{equation}
\label{eq:th_eq10}
    \parallel \nabla F(\mathbf{w}_t) \parallel^2 \geq 2\mu(F(\mathbf{w}_t)-F^*).
\end{equation}
\end{small}

By subtracting $F^*$ from both sides of the inequality and based on \eqref{eq:th_eq10} yields
\vspace{-1.5mm}
\begin{small}
\begin{equation}
\label{eq:th_eq11}
    \begin{aligned}
        E[F(\mathbf{w}_{t+1})-F^*] &\leq (1-2\mu \eta + 2L\mu \eta^2)E[F(\mathbf{w}_t)-F^*]\\
        &+ 2L\eta^2 \Big (\frac{\sigma^2_0}{\gamma^s_\text{min}} \frac{ K}{\sum_kD_kq_k}+ \zeta^2\sum^K_{k=1}\bar{\rho}_k^2 \Big).
    \end{aligned}
\end{equation}
\end{small}

Applying \eqref{eq:th_eq11} recursively, we have
\vspace{-1.5mm}
\begin{small}
\begin{equation}
\label{eq:th_eq12}
    \begin{aligned}
        E[F(\mathbf{w}_{t})-F^*] &\leq A^tE[F(\mathbf{w}_0)-F^*]\\
        &+ \frac{2L\eta^2(1-A^t)}{1-A}\Big (\frac{\sigma^2_0}{\gamma^s_\text{min}} \frac{ K}{\sum_kD_kq_k}+ \zeta^2\sum^K_{k=1}\bar{\rho}_k^2 \Big),
    \end{aligned}
\end{equation}
where $A=1-2\mu \eta + 2L\mu \eta^2$ and $F^*$ is the minimal value of $F(\mathbf{w})$.
\end{small}

\bibliographystyle{IEEEtran}
\bibliography{reference}

\end{document}